\documentclass[11pt,runningheads,envcountsame,orivec]{llncs}
\usepackage[a4paper,centering,margin=1in]{geometry}

\input{preamble}
\usepackage{tikz}
\usetikzlibrary{
    arrows.meta,backgrounds,calc,
    decorations,intersections,
    positioning,shapes.geometric,
}

\tikzset{
    spacing/.style={inner sep=0.8pt},
    every label/.style={thin,inner sep=0.6pt,font=\footnotesize},
    every pin/.style={thin,inner sep=1pt,font=\footnotesize},
    point/.style={draw,semithick,fill,diamond,inner sep=0pt,minimum size=5.7pt},
    round point/.style={draw,semithick,fill,circle,inner sep=0pt,minimum size=4.5pt},
    vector/.style={draw,very thick,-{Triangle[fill=white,line width=1pt,width=4.5pt,length=6.5pt]}},
    main line/.style={draw,very thick},
    curve/.style={main line,miter limit=10,vertex/.style={round point}},
    valley/.style={draw=UmiOrange,ultra thick},
    matching/.style={draw=UmiSkyblue,very thick},
    matching path/.style={draw,semithick,double=UmiSkyblue,double distance=1.6pt},
    matching point/.style={point,fill=UmiSkyblue},
    faint fill/.style={fill={#1!10!white}},
    split fill/.style 2 args={
        fill={#1},path picture={
            \path[fill={#2}] (path picture bounding box.south west) -| (path picture bounding box.north east) -- cycle;
        },
    },
    axis/.style={draw,semithick,black!60!white,-{Straight Barb}},
    axes/.style={axis,{Straight Barb}-{Straight Barb}},
    tick/.style={axis,-,every node/.style={font=\footnotesize}},
    interval/.style={
        draw,semithick,black!60!white,|-|,shorten >=-0.3pt,shorten <=-0.3pt,
        every node/.style={fill=white,font=\footnotesize},
    },
    interval e/.style={interval,xshift=10pt,every node/.append style={inner xsep=0pt,inner ysep=2pt}},
    interval h/.style={interval,yshift={#1},every node/.append style={inner xsep=2pt,inner ysep=0pt}},
    interval n/.style={interval h=10pt},interval s/.style={interval h=-10pt},
    util grey/.style={draw=black!35!white,semithick},
    util dash/.style={dash pattern=on 3pt off 2pt,dash expand off},
}

\definecolor{UmiGreen}{rgb}{0.00,0.72,0.62}
\definecolor{UmiOrange}{rgb}{0.95,0.55,0.00}
\definecolor{UmiBlue}{rgb}{0.00,0.45,0.70}
\definecolor{UmiSkyblue}{rgb}{0.35,0.70,0.90}
\definecolor{UmiYellow}{rgb}{0.95,0.90,0.20}
\definecolor{UmiVermillion}{rgb}{0.80,0.40,0.00}

\newcommand{\apxhatch}[2][1]{
    \begin{scope}
        \clip (0,0) rectangle (#1,1);

        \pgfmathparse{1/(#2+1)}\edef\k{\pgfmathresult}
        \pgfmathparse{(#1+1)*(#2+1)-1}\edef\n{\pgfmathresult}

        \foreach \i in {2,...,\n} {
            \pgfmathparse{\i-0.5}\edef\j{\pgfmathresult}
            \path[draw=black!10!white] (\j*\k-1,0) -- (\j*\k,1);
            \path[draw=black!10!white] (0,\j*\k) -- (\j*\k,0);
        }
    \end{scope}
}

\newcommand{\examplesegments}{
    \coordinate (px) at (0,3);
    \coordinate (py) at (5,5);
    \path[curve]
        (px) node[vertex,label={above left:$p_x$}] {} -- (py) node[vertex,label={above right:$p_y$}] {};

    \coordinate (qx) at (0.75,0);
    \coordinate (qy) at (6,0.75);
    \path[curve,vertex/.append style={rectangle}]
        (qx) node[vertex,label={below left:$q_x$}] {} -- (qy) node[vertex,label={below right:$q_y$}] {};
}

\hypersetup{
    pdftitle={A Constant-Factor Approximation for Continuous Dynamic Time Warping in 2D},
    pdfauthor={Kevin Buchin, Maike Buchin, Jan Erik Swiadek, Sampson Wong},
}

\begin{document}
    \title{A Constant-Factor Approximation for\texorpdfstring{\\}{} Continuous Dynamic Time Warping in 2D}
    \author{
        Kevin Buchin\inst{1}
        \and
        Maike Buchin\inst{2}
        \and
        Jan Erik Swiadek\inst{2}
        \and
        Sampson Wong\inst{3}
    }
    \institute{
        Technical University Dortmund, Germany\\ \email{kevin.buchin@tu-dortmund.de}
        \and
        Ruhr University Bochum, Germany\\ \email{\{maike.buchin,jan.swiadek\}@rub.de}
        \and
        University of Copenhagen, Denmark\\ \email{sampson.wong123@gmail.com}
    }
    \titlerunning{A Constant-Factor Approximation for CDTW in 2D}
    \authorrunning{K. Buchin, M. Buchin, J. E. Swiadek, S. Wong}
    \maketitle

    \begin{abstract}
        Continuous Dynamic Time Warping (CDTW) is a robust similarity measure for polygonal curves that has recently found a variety of applications.
        Despite its practical use, not much is known about the algorithmic complexity of computing it in 2D, especially when one requires either an exact solution or strong approximation guarantees.
        We fill this gap by introducing a $5$-approximation algorithm with running time $O(n^5)$ under the $1$-norm.

        This is the first constant-factor approximation for 2D CDTW with polynomial running time.
        We extend our algorithm to all polygonal norms on $\mathbb{R}^2$, which we subsequently use in order to achieve a $(5+\varepsilon)$-approximation with time complexity $O(n^5 / \varepsilon^{1/2})$ for CDTW in 2D under any fixed norm.
        The latter result in particular includes the usual Euclidean $2$-norm.

        \keywords{Continuous Dynamic Time Warping \and curve similarity \and geometric approximation}
    \end{abstract}

    \section{Introduction}

There exist several desirable properties for curve similarity measures \cite{SuLZZZ2020,TaoBSBSPLPTD2021}, two of which are tolerance to outliers and robustness to different sampling rates.
As the Fréchet distance fails to provide the former, whereas Dynamic Time Warping struggles with the latter, these popular measures have drawbacks for applications.
Continuous Dynamic Time Warping (CDTW) is an alternative that combines both abovementioned properties by minimising a path integral of distances between continuously and monotonically matched points on two polygonal curves.

Our considered CDTW formulation \cite{MahesSS2018,Klare2020,BrankBKNPW2020,BuchiNW2025,Brank2022,BuchiBSW2026} originates from a summed version of the Fréchet distance introduced in \cite{Buchi2007}.
There exist other related definitions \cite{SerraB1994,SerraB1995,MunicP1999,BrakaPSW2005,EfratFV2007,Har-PRR2025}, and we refer to \cite[Section~1]{BuchiNW2025} as well as \cite[Section~1.1]{BuchiBSW2026} for brief overviews of these.
On the practical~side, the robustness of CDTW and its ability to yield high-quality solutions across different application domains have been observed in experiments~\cite{BrakaPSW2005,EfratFV2007,BrankBKNPW2020}.
Also, reasonable running times on realistic data have been reported \cite{BrankBKNPW2020,Brank2022}.
On the theoretical~side, however, CDTW still lacks algorithms that guarantee both a good approximation factor and an efficient running time in 2D and beyond.

Many existing approaches are approximations or heuristics that proceed via discretising the input curves \cite{BrakaPSW2005,EfratFV2007,MahesSS2018,Klare2020,BrankBKNPW2020,Har-PRR2025}.
These include a generic additive approximation algorithm~\cite{Klare2020,BrankBKNPW2020}, and a pseudo-polynomial-time $(1 + \varepsilon)$-approximation for 2D CDTW under the Euclidean \mbox{$2$-norm}~\cite{MahesSS2018}.
The latter has running time $O(\zeta^4 n^4 / \varepsilon^2 \cdot \log(\zeta n / \varepsilon))$, where $n$ is the complexity and $\zeta$ is the spread of the given polygonal curves.
That means the running time depends on the maximum ratio $\zeta$ of a pair of curve segment lengths.
Such dependencies are typical for discretisation approaches.

A different type of approach, which enables exact computations, propagates functions arising from integration within a dynamic program \cite{BuchiBW2009,Klare2020,BuchiNW2025,Brank2022,BuchiBSW2026}.
This results in a polynomial running~time of $O(n^5)$ for CDTW in 1D \cite{BuchiNW2025}.
However, in 2D the algorithmic complexity of computing CDTW is as yet unknown.
On the one hand, it has been shown that 2D CDTW cannot even be computed exactly under the $2$-norm when restricted to algebraic operations \cite{BuchiBSW2026}.
On the other hand, function propagation approaches have only been made to work with polygonal norms like the $1$-norm \cite{Brank2022,BuchiBSW2026}.
These give piecewise linear integrands and piecewise quadratic cost functions similar to 1D.

By converging to the unit disk, polygonal norms provide a $(1 + \varepsilon)$-approximation~for~CDTW under the $2$-norm without a dependency on the curves' spread.
However, the 2D dynamic program may involve more complicated propagation patterns \cite[Section~4.2]{BuchiBSW2026}.
Those have so far prevented any 2D generalisation of the quite technical polynomial-time result that was achieved in 1D.

\paragraph*{Our Contributions.}

We show that we can circumvent the 2D propagation patterns that potentially have an exponential complexity, in exchange for an approximation factor $\beta \leq 5$.
This relies on a new bound for certain CDTW integrals, which we establish in \autoref{subsec:building-block}.
It is open whether the bound is tight, and any improvement to $\beta$ would transfer to the following results.
Applying this building block, we develop a CDTW algorithm that is outlined in \autoref{subsec:algorithm-outline} and then analysed in \autoref{sec:algorithm-analysis}.
After proving its approximation factor, we bound its time complexity by combining ideas from~\cite{BuchiNW2025} with novel insights.
For CDTW in $(\mathbb{R}^2,\tnorm[1])$ we achieve a $5$-approximation with running time $O(n^5)$.
This is the first polynomial-time constant-factor result for 2D CDTW.

We subsequently extend our algorithm from the $1$-norm $\tnorm[1]$ to the class of polygonal norms.
This is based on the framework of~\cite{BuchiBSW2026}, which we generalise to facilitate approximations for not only the $2$-norm $\tnorm[2]$ but also any other norm on $\mathbb{R}^2$.
To that end, we exploit that any fixed~norm~$\tnorm$ on $\mathbb{R}^d$ is $(1 + \varepsilon)$-approximated by a polyhedral norm of complexity $O(\varepsilon^{(1-d)/2})$.
Here, fixed means that we treat the description complexity of $\tnorm$, including the dimension $d$, as a constant.

Putting things together, we obtain an approximation algorithm with factor $5 + \varepsilon$ and running time $O(n^5 / \varepsilon^{1/2})$ for 2D CDTW under any fixed norm and for any $\varepsilon > 0$.
In the imbalanced case, where one of the given polygonal curves has some smaller complexity $m \in o(n)$, the term $n^5$ in our time bounds may be replaced by $n^4 m$.
Note that the guaranteed bounds might be pessimistic.
In particular, even if $\beta = 5$ is a tight bound for the approximated CDTW integrals, it is unclear whether there are any curves for which our algorithm approximates within factor close to $\beta$.

On a technical level, we expand the toolbox for CDTW algorithms by demonstrating that an integration-based building block and a tailored function propagation scheme can enable CDTW approximations without discretisations of any form.
Combining our techniques with the approach of approximating the Euclidean $2$-norm via polygonal norms yields polynomial-time computations for 2D CDTW under $\tnorm[2]$ without any significant loss of accuracy for the first time.
Even though polygonal norms have already been used in \cite{BuchiBSW2026}, no attempt at a fine-grained~running time analysis was made there, after highlighting the obstacle that we deal with approximatively here.

Furthermore, our running time analysis is more streamlined and insightful than that from~\cite{BuchiNW2025}.
We avoid a large case distinction by identifying core principles behind the complexity of function propagations, which generalise beyond some properties unique to 1D.
In addition to advancing the understanding of propagation patterns, this enables an improvement upon previous algorithms:
We compute lower envelopes of quadratic pieces in a single pass instead of multiple passes.

\section{Preliminaries}

A polygonal curve $P$ in a normed real vector space $(\mathbb{R}^d, \tnorm)$ consists of $n \in \mathbb{N}$ consecutive line segments induced by a sequence $\langle p_0,\dotsc,p_n \rangle$ of vertices, where $p_{i} \neq p_{i-1}$ for all $i \in \{1,\dotsc,n\}$.
We write $P_{\leq i} := \langle p_0,\dotsc,p_i \rangle$, where $i \in \{0,\dotsc,n\}$, for the prefix subcurves of $P$.
The following definitions are based on \cite{BuchiBSW2026}, which provides a robust CDTW formulation under any norm.

\begin{definition}[{\cite[Definition~1]{BuchiBSW2026}}]
    \label{def:cdtw}
    Let $P,Q$ be two polygonal curves in $(\mathbb{R}^d, \tnorm)$.
    The measure \emph{Continuous Dynamic Time Warping (CDTW)} of $P,Q$ under $\tnorm$ is defined by
    \[
        \mathrm{cdtw}_{\| \cdot \|}(P,Q)
        :=
        \inf_{(f,g) \in \Pi_{[0,1]}(P) \times \Pi_{[0,1]}(Q)}
        \int_0^1
        \| f(t) - g(t) \| \cdot
        \left\| \begin{pmatrix}
            \| f'(t) \| \\
            \| g'(t) \|
        \end{pmatrix} \right\|_1
        \dt
        \text{,}
    \]
    where the sets $\Pi_{[a,b]}(P)$ and $\Pi_{[a,b]}(Q)$ contain all piecewise continuously differentiable functions defined on the interval $[a,b]$ that monotonically parametrise $P$ or $Q$ respectively.
\end{definition}

The arc length of $P$ under $\tnorm$ is $\| P \| := \sum_{i=1}^n \| p_i - p_{i-1} \|$.
Because the arc length is invariant to reparametrisation, we have $\int_{a}^{b} \| f'(t) \| \dt = \| P \|$ for all $f \in \Pi_{[a,b]}(P)$.
Moreover, we denote the arc length parametrisation of $P$ with constant speed~$1$ under $\tnorm$ by $P_{\| \cdot \|} \colon [0, \| P \|] \to \mathbb{R}^d$.

\begin{definition}[{\cite[Definitions~2 and~12]{BuchiBSW2026}}]
    \label{def:parameter-space-and-borders}
    The \emph{parameter space} of two polygonal curves $P,Q$ under $\tnorm$ is $[0,\| P \|] \times [0, \| Q \|]$.
    Each segment pair of $P = \langle p_0,\dotsc, p_n \rangle$ and $Q = \langle q_0,\dotsc,q_m \rangle$ is associated with a \emph{cell} $C_{i,j} := [\| P_{\leq i-1} \|, \| P_{\leq i} \|] \times [\| Q_{\leq j-1} \|, \| Q_{\leq j} \|]$ in their parameter space, where $(i,j) \in \{1,\dotsc,n\} \times \{1,\dotsc,m\}$.
    The \emph{north/east/south/west border} of $C_{i,j}$ parametrises along that side of $C_{i,j}$, e.g.\ its north border is $[\| P_{\leq i-1} \|, \| P_{\leq i} \|] \to \mathbb{R}^2, t \mapsto (t,\| Q_{\leq j} \|)^{\mathsf{T}}$.
\end{definition}

The monotone matchings of $P$ and $Q$ correspond to monotone paths in their parameter space.
Let $\sigma := \| P \| + \| Q \|$, and define $\Gamma_{\| \cdot \|}(P,Q)$ as the set of all functions $\gamma \colon [0,\sigma] \to [0, \| P \|] \times [0, \| Q \|]$ such that there are parametrisations $(f,g) \in \Pi_{[0,\sigma]}(P) \times \Pi_{[0,\sigma]}(Q)$ satisfying $\gamma_1(s) = \int_{0}^{s} \| f'(t) \| \dt$ and $\gamma_2(s) = \int_{0}^{s} \| g'(t) \| \dt$ as well as $\| \gamma'(s) \|_1 = \| f'(s) \| + \| g'(s) \| = 1$ for all $s \in [0, \sigma]$.

\begin{definition}[{\cite[Definition~3]{BuchiBSW2026}}]
    \label{def:path-and-cost}
    Let $x,y \in [0, \| P \|] \times [0, \| Q \|]$ be two points with $x \preceq y$, i.e.\ $x_1 \leq y_1$ and $x_2 \leq y_2$, for polygonal curves $P,Q$.
    An \emph{$(x,y)$-path} is a restriction $\widehat{\gamma} := \gamma|_{[\| x \|_1, \| y \|_1]}$ of any function $\gamma \in \smash{\Gamma_{\| \cdot \|}}(P,Q)$ with $\gamma(\| x \|_1) = x$ and $\gamma(\| y \|_1) = y$.
    It is called \emph{optimal} for $P,Q$ under $\tnorm$ if it minimises the definite integral $\mathrm{cost}_{\| \cdot \|}(\widehat{\gamma}) := \int_{\| x \|_1}^{\| y \|_1} \| P_{\| \cdot \|}(\widehat{\gamma}_1(t)) - Q_{\| \cdot \|}(\widehat{\gamma}_2(t)) \| \dt$ among all $(x,y)$-paths.
    We write $\mathrm{opt}_{\| \cdot \|}(x,y)$ for the cost of such an optimal $(x,y)$-path.
\end{definition}

We have $\mathrm{cdtw}_{\| \cdot \|}(P,Q) = \mathrm{opt}_{\| \cdot \|}(\mathbf{0},(\| P \|, \| Q \|)^{\mathsf{T}})$ by construction, where $\mathbf{0}$ denotes the origin.
Given some cell border $\mathcal{B} \colon \mathrm{dom}(\mathcal{B}) \to \mathbb{R}^2$, we define its \emph{optimum function} $\mathrm{opt}_{0,\mathcal{B}} \colon \mathrm{dom}(\mathcal{B}) \to \mathbb{R}_{\geq 0}$ by $\mathrm{opt}_{0,\mathcal{B}}(t) := \mathrm{opt}_{\| \cdot \|}(\mathbf{0},\mathcal{B}(t))$ for $t \in \mathrm{dom}(\mathcal{B})$.
In a cell under $\tnorm$ we further define for each pair of south/west \emph{input border} $\mathcal{A}$ and north/east \emph{output border} $\mathcal{B}$ a function $\mathrm{opt}_{\mathcal{A},\mathcal{B}}$ by
\[
    \mathrm{opt}_{\mathcal{A},\mathcal{B}}(s,t) := \mathrm{opt}_{\| \cdot \|}(\mathcal{A}(s),\mathcal{B}(t))
    \quad
    \text{for }
    (s,t) \in \mathrm{dom}(\mathcal{A}) \times \mathrm{dom}(\mathcal{B})
    \text{ with }
    \mathcal{A}(s) \preceq \mathcal{B}(t)
    \text{.}
\]

The optimal paths realising $\mathrm{opt}_{\mathcal{A},\mathcal{B}}$ are characterised by \autoref{thm:optimal-paths}, and possible resulting shapes of these optimal paths within different cells are illustrated in \autoref{fig:optimal-path-types}.

\begin{theorem}[{\cite[Section~2]{BuchiBSW2026}}]
    \label{thm:optimal-paths}
    Let $P,Q$ be two polygonal curves in $(\mathbb{R}^2,\tnorm)$, and let $C$ be a parameter space cell.
    There exists a line $\ell \subseteq \mathbb{R}^2$ of positive slope such that for any choice of two points $x,y \in C$ with $x \preceq y$ the following $(x,y)$-path $\gamma$ is optimal for $P,Q$ under $\tnorm$.
    \begin{alphaenumerate*}
        \item
        If $\ell$ intersects the bounding box $\Xi := [x_1,y_1] \times [x_2,y_2]$ of $x$ and $y$, then $\gamma$ traces line segments from $x$ to $\xi_x$ to $\xi_y$ to $y$, where $\xi_x,\xi_y \in \ell \cap \Xi$ share a coordinate with $x,y$ respectively.

        \item
        If $\ell$ does not intersect the bounding box $\Xi$, then $\gamma$ traces line segments from $x$ to $\xi$ to $y$, where the single bending point $\xi \in \{ (x_1,y_2)^{\mathsf{T}}, (y_1,x_2)^{\mathsf{T}} \}$ is the point closest to $\ell$ within $\Xi$.
    \end{alphaenumerate*}
\end{theorem}

\begin{figure}[H]
    \centering
    \begin{tikzpicture}[x=2.75cm,y=2.75cm,trim left=0cm,trim right={5.6 * 2.75cm}]
        \begin{scope}
            \path[valley] (2/3,0) -- (1,1/3);
            \path[main line] (0,0) rectangle (1,1);

            \path[matching path] (1/6,0) -- (1/3,0) -- (1/3,1);
            \node[matching point] at (1/6,0) {};
            \node[matching point] at (1/3,1) {};

            \path[matching path] (0.5,0) -- (2/3,0) -- (5/6,1/6) -- (5/6,1);
            \node[matching point] at (0.5,0) {};
            \node[matching point] at (5/6,1) {};
        \end{scope}

        \begin{scope}[shift={(1.15,0)}]
            \path[valley] (0,0) -- (1,1);
            \path[main line] (0,0) rectangle (1,1);

            \path[matching path] (0.25,0) -- (0.25,1);
            \node[matching point] at (0.25,0) {};
            \node[matching point] at (0.25,1) {};

            \path[matching path] (0.5,0) -- (0.5,0.5) -- (0.75,0.75) -- (0.75,1);
            \node[matching point] at (0.5,0) {};
            \node[matching point] at (0.75,1) {};
        \end{scope}

        \begin{scope}[shift={(2.3,0)}]
            \path[valley] (0,2/3) -- (1/3,1);
            \path[main line] (0,0) rectangle (1,1);

            \path[matching path] (1/6,0) -- (1/6,5/6) -- (1/3,1) -- (0.5,1);
            \node[matching point] at (1/6,0) {};
            \node[matching point] at (0.5,1) {};

            \path[matching path] (2/3,0) -- (2/3,1) -- (5/6,1);
            \node[matching point] at (2/3,0) {};
            \node[matching point] at (5/6,1) {};
        \end{scope}

        \begin{scope}[shift={(3.45,0)}]
            \path[valley] (0,0.5) -- node[above left,spacing] {\color{UmiOrange}$\ell$} (0.5,1);
            \path[main line] (0,0) rectangle (1,1);

            \path[matching path] (0.5,0) -- (0.5,2/3) -- (1,2/3);
            \node[matching point] at (0.5,0) {};
            \node[matching point] at (1,2/3) {};

            \path[matching path] (1,0) -- (1,1/3);
            \node[matching point] at (1,0) {};
            \node[matching point] at (1,1/3) {};
        \end{scope}

        \begin{scope}[shift={(4.6,0)}]
            \path[valley] (1/3,0) -- (1,2/3);
            \path[main line] (0,0) rectangle (1,1);

            \path[matching path] (1/6,0) -- (1/3,0) -- (0.5,1/6) -- (1,1/6);
            \node[matching point] at (1/6,0) {};
            \node[matching point] at (1,1/6) {};

            \path[matching path] (2/3,0) -- (2/3,1/3) -- (5/6,0.5) -- (1,0.5);
            \node[matching point] at (2/3,0) {};
            \node[matching point] at (1,0.5) {};
        \end{scope}
    \end{tikzpicture}
    \caption{Optimal paths from different cells' south borders to their output borders}
    \label{fig:optimal-path-types}
\end{figure}
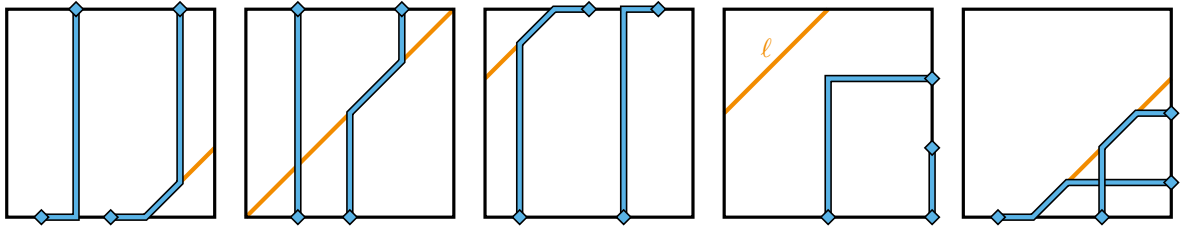

Note that \autoref{thm:optimal-paths} relies on the usage of norms in \autoref{def:cdtw}:
The equipped norm~$\tnorm$ measures both distances between curve points and speeds of curve parametrisations \cite{BuchiBSW2026}, while the $1$-norm $\tnorm[1]$ combines these speeds within the CDTW path integral \cite{Buchi2007,MahesSS2018,BuchiNW2025,Brank2022,BuchiBSW2026}.
Moreover, \autoref{thm:optimal-paths} only holds in 1D and 2D.
See \autoappref{app:3d-example} for a 3D counterexample.

\section{Towards a Polynomial Bound}

The idea behind function propagation approaches for exact CDTW algorithms is propagating costs of optimal paths through the parameter space cells and successively computing optimum functions in a dynamic program.
After $\mathrm{opt}_{0,\mathcal{A}}$ has been computed for each input border $\mathcal{A}$ of some cell $C$, one can use \autoref{thm:optimal-paths} to determine $\mathrm{opt}_{0,\mathcal{B}}$ for each output border $\mathcal{B}$ of $C$.

If the CDTW integrands are piecewise linear, as it is the case in 1D \cite{BuchiNW2025} or under polygonal norms in 2D such as the $1$-norm or $\infty$-norm \cite{Brank2022,BuchiBSW2026}, then all optimum functions of borders are piecewise quadratic (cf.~\cite[Lemma~14]{BuchiBSW2026}).
The central challenge for the running time analysis lies in bounding the total number of quadratic pieces over all optimum functions \cite{SerraB1994,MunicP1999,BuchiNW2025,BuchiBSW2026}.

\subsection{Building Block}
\label{subsec:building-block}

In general, the possible propagation patterns and their complexity in terms of how the number of quadratic pieces is growing are still not well understood.
The polynomial bound for the exact~1D algorithm relies heavily on the fact that the quadratic pieces on the output borders of a cell can be assigned monotonically to \emph{parent} pieces on the given cell's input borders (cf.~\cite[Lemma~13 and Definition~18]{BuchiNW2025}).
However, this is insufficient on its own already in 2D under the $1$-norm:

Certain optimal paths originating from \autoref{thm:optimal-paths}b may hypothetically cause an exponential growth of propagated pieces in 2D, even though there is a monotone assignment to parent pieces.
In short, the obstacle for bounding the number of pieces occurs when the paths' bending points move freely within the cell.
This yet unresolved issue was highlighted in \cite[Section~4.2]{BuchiBSW2026}.

To tackle the absence of means for finding general tight bounds, our approach instead aims to ensure that the complexity created by such paths during the propagation of path costs in the dynamic program remains low.
We henceforth call them \emph{unhappy} paths within this paper.

\begin{definition}
    \label{def:happy-and-unhappy}
    Let $\gamma$ be a polygonal $(x,y)$-path as in \autoref{thm:optimal-paths}.
    If each bending point of~$\gamma$ lies on $\ell$ or on the boundary of the cell $C$, then $\gamma$ is \emph{happy} with respect to $\ell$.
    Else, $\gamma$ is \emph{unhappy}.
\end{definition}

See the fourth cell of \autoref{fig:optimal-path-types} for an unhappy path.
As our first contribution, we show that the cost of any (unhappy) optimal $(x,y)$-path $\gamma$ from \autoref{thm:optimal-paths}b is approximated within factor at most~$5$ by choosing the worse of the two possible options for the single bending point $\xi$ of~$\gamma$.
This changes whether a path first travels vertically and then horizontally or vice versa, as shown in \autoref{fig:approximable-paths}.
Note that the paths from \autoref{thm:optimal-paths}a are always happy by definition.

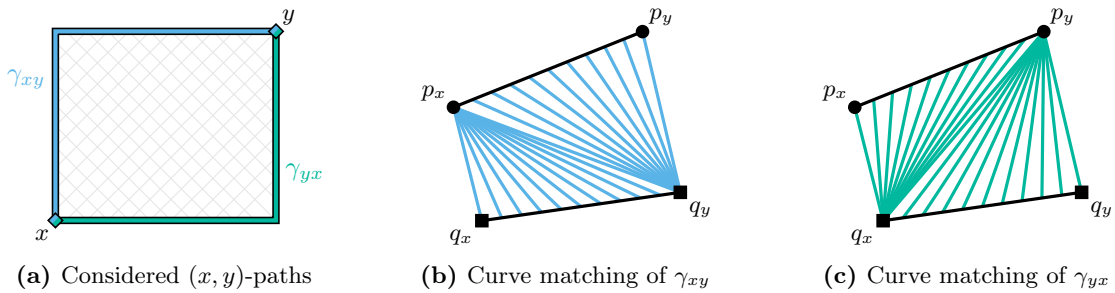
\begin{figure}[H]%
    \centering%
    \begin{subfigure}{0.3\linewidth}
        \centering
        \begin{tikzpicture}[x=2.5cm,y=2.5cm]
            \apxhatch[7/6]{8}

            \path[matching path] (0,0) -- node[pos=0.75,left] {\color{UmiSkyblue}$\gamma_{xy}$} (0,1) -- (7/6,1);
            \path[matching path,double=UmiGreen] (0,0) -- (7/6,0) -- node[pos=0.25,right] {\color{UmiGreen}$\gamma_{yx}$} (7/6,1);

            \node[point,split fill={UmiSkyblue}{UmiGreen},label={[text height=1.15ex]below left:$x$\vphantom{$q_x$}}] at (0,0) {};
            \node[point,split fill={UmiSkyblue}{UmiGreen},label={above right:$y$}] at (7/6,1) {};
        \end{tikzpicture}
        \subcaption{Considered $(x,y)$-paths}
    \end{subfigure}%
    \hspace*{0.1\linewidth/3}%
    \begin{subfigure}{0.3\linewidth}
        \centering
        \begin{tikzpicture}[x={2.5cm / 5},y={2.5cm / 5}]
            \examplesegments

            \begin{scope}[on background layer]
                \def\n{10}
                \foreach \i in {0,...,\n} {
                    \path[matching] (px) -- ($(qx)!{\i/\n}!(qy)$);
                    \ifnum\i>0 \path[matching] (qy) -- ($(px)!{\i/\n}!(py)$); \fi
                }
            \end{scope}
        \end{tikzpicture}
        \subcaption{Curve matching of $\gamma_{xy}$}
    \end{subfigure}%
    \hspace*{0.1\linewidth/3}%
    \begin{subfigure}{0.3\linewidth}
        \centering
        \begin{tikzpicture}[x={2.5cm / 5},y={2.5cm / 5}]
            \examplesegments

            \begin{scope}[on background layer]
                \def\n{10}
                \foreach \i in {0,...,\n} {
                    \path[matching,UmiGreen] (qx) -- ($(px)!{\i/\n}!(py)$);
                    \ifnum\i>0 \path[matching,UmiGreen] (py) -- ($(qx)!{\i/\n}!(qy)$); \fi
                }
            \end{scope}
        \end{tikzpicture}
        \subcaption{Curve matching of $\gamma_{yx}$}
    \end{subfigure}%
    \caption{Approximable paths and their corresponding curve matchings}%
    \label{fig:approximable-paths}%
\end{figure}

Our technical proof of the following lemma relies -- similar to \autoref{thm:optimal-paths} -- on \autoref{def:cdtw} and \autoref{def:path-and-cost} employing the equipped norm~$\tnorm$ for distances between curve points as well as speeds and arc lengths on each curve.
In particular, this allows for a versatile use of the triangle inequality and symmetry of the metric $(p,q) \mapsto \| p - q \|$ induced by the norm $\tnorm$ on $\mathbb{R}^d$.

\begin{lemma}
    \label{thm:apx-lemma}
    Let $x,y \in C$ be two points with $x \preceq y$, where $C$ is a parameter space cell of polygonal curves $P,Q$ in $(\mathbb{R}^d,\tnorm)$.
    Consider the $(x,y)$-path $\gamma_{xy}$ tracing line segments from $x$ to $(x_1,y_2)^{\mathsf{T}}$ to $y$, and the $(x,y)$-path $\gamma_{yx}$ tracing line segments~from $x$ to $(y_1,x_2)^{\mathsf{T}}$ to $y$.
    The costs of these paths under $\tnorm$ satisfy the bounds $1/5 \cdot \mathrm{cost}_{\| \cdot \|}(\gamma_{xy}) \leq \mathrm{cost}_{\| \cdot \|}(\gamma_{yx}) \leq 5 \cdot \mathrm{cost}_{\| \cdot \|}(\gamma_{xy})$.
\end{lemma}

\begin{proof}
    We will show only the second inequality, as the first one is entirely symmetric.
    In order to make the notation of this proof more concise, we write $p_s := P_{\| \cdot \|}(s) \in \mathbb{R}^d$ and $q_{t} := Q_{\| \cdot \|}(t) \in \mathbb{R}^d$ for parameters $(s, t) \in [x_1,y_1] \times [x_2,y_2]$ as well as $\overline{p q} := \| p - q \|$ for points $p,q \in \smash{\mathbb{R}^d}$.

    Because the paths $\gamma_{xy}$ and $\gamma_{yx}$ trace horizontal and vertical line segments, \autoref{def:path-and-cost} implies that their cost integrands on these respective segments correspond to $s \mapsto \overline{p_s q_t}$, where $s \in [x_1,y_1]$ varies and $t \in \{x_2,y_2\}$ is fixed, and to $t \mapsto \overline{p_s q_t}$, where $t \in [x_2,y_2]$~varies and $s \in \{x_1,y_1\}$ is fixed.
    We thus also use the shorthands $p_x := p_{x_1}$, $p_y := p_{y_1}$, $q_x := q_{x_2}$, and $q_y := q_{y_2}$.
    By applying the triangle inequality and symmetry that are induced by the norm $\tnorm$, we then obtain
    \begin{align*}
        \mathrm{cost}_{\| \cdot \|}(\gamma_{yx}) - \mathrm{cost}_{\| \cdot \|}(\gamma_{xy})
        &= \int_{x_1}^{y_1} \overline{p_s q_x} \ds + \int_{x_2}^{y_2} \overline{p_y q_t} \dt - \int_{x_2}^{y_2} \overline{p_x q_t} \dt - \int_{x_1}^{y_1} \overline{p_s q_y} \ds \\[3pt]
        &= \int_{x_1}^{y_1} \overline{p_s q_x} - \overline{p_s q_y} \ds + \int_{x_2}^{y_2} \overline{p_y q_t} - \overline{p_x q_t} \dt \\[3pt]
        &\leq \int_{x_1}^{y_1} \overline{q_x q_y} \ds + \int_{x_2}^{y_2} \overline{p_x p_y} \dt
        = 2 \cdot \overline{p_x p_y} \cdot \overline{q_x q_y}
        \text{.}
    \end{align*}

    The final equality holds because the integrands do not depend on the variables $s$ and $t$ anymore, while the integration intervals are of length $y_1 - x_1 = \overline{p_x p_y}$ and $y_2 - x_2 = \overline{q_x q_y}$ respectively.
    Here, note that we have $y_1 - x_1 = \overline{p_x p_y}$ due to $x,y \in C$ for the given cell $C$, which means that the two points $p_x$ and $p_y$ lie on a single curve segment of $P$.
    We similarly have $y_2 - x_2 = \overline{q_x q_y}$.
    Intuitively, the cost difference between $\gamma_{yx}$ and $\gamma_{xy}$ is at most two times the bounding box area of $x$ and $y$.
    Together with $0 \leq (\overline{p_x p_y} - \overline{q_x q_y})^2 = \overline{p_x p_y}^2 + \overline{q_x q_y}^2 - 2 \cdot \overline{p_x p_y} \cdot \overline{q_x q_y}$, we therefore arrive at
    \[
        \mathrm{cost}_{\| \cdot \|}(\gamma_{yx})
        \leq \mathrm{cost}_{\| \cdot \|}(\gamma_{xy}) + 2 \cdot \overline{p_x p_y} \cdot \overline{q_x q_y}
        \leq \mathrm{cost}_{\| \cdot \|}(\gamma_{xy}) + \overline{p_x p_y}^2 + \overline{q_x q_y}^2
        \text{,}
    \]
    so it remains to show $\overline{p_x p_y}^2 + \overline{q_x q_y}^2 \leq 4 \cdot \mathrm{cost}_{\| \cdot \|}(\gamma_{xy})$ in order to get the claimed result.

    For this, we consider $\int_{x_2}^{y_2} \overline{p_x q_t} \dt$, which is the first term of $\mathrm{cost}_{\| \cdot \|}(\gamma_{xy})$, and pick a $\lambda \in [x_2,y_2]$ such that both $\overline{q_x q_\lambda} \leq \overline{p_x q_x}$ and $\overline{q_\lambda q_y} \leq \overline{p_x q_y}$ hold.
    This is feasible as follows:
    If $\overline{q_x q_y} \leq \overline{p_x q_x}$~holds, then $\lambda := y_2$ with $q_\lambda = q_y$ works.
    Else, we have $\overline{q_x q_y} \geq \overline{p_x q_x}$ and can pick $\lambda$ with $\overline{q_x q_\lambda} = \overline{p_x q_x}$,~which implies $\overline{q_\lambda q_y} = \overline{q_x q_y} - \overline{q_x q_\lambda} = \overline{q_x q_y} - \overline{p_x q_x} \leq \overline{p_x q_y}$ via $x,y \in C$ and the triangle inequality.

    We now split the interval $[x_2,\lambda]$ into $k \in \mathbb{N}$ subintervals, each of length $\frac{1}{k} \cdot (\lambda - x_2) = \frac{1}{k} \cdot \overline{q_x q_\lambda}$.
    For every value $t$ on the $i$-th subinterval this gives $\overline{q_x q_t} \leq \frac{i}{k} \cdot \overline{q_x q_\lambda}$, so together we have
    \[
        \overline{p_x q_t} \geq \overline{p_x q_x} - \overline{q_x q_t} \geq \overline{q_x q_\lambda} - \frac{i}{k} \cdot \overline{q_x q_\lambda} = \frac{k-i}{k} \cdot \overline{q_x q_\lambda}
        \quad
        \text{if }
        t - x_2 \in \left[ \frac{i-1}{k} \cdot \overline{q_x q_\lambda}, \frac{i}{k} \cdot \overline{q_x q_\lambda} \right]
        \text{,}
    \]
    which bounds the integrand from below.
    Summing this over all $i \in \{1,\dotsc,k\}$ yields
    \[
        \int_{x_2}^{\lambda} \overline{p_x q_t} \dt
        \geq \frac{1}{k} \cdot \overline{q_x q_\lambda} \cdot \sum_{i=1}^{k} \frac{k - i}{k} \cdot \overline{q_x q_\lambda}
        = \frac{1}{k} \cdot \overline{q_x q_\lambda} \cdot \frac{(k-1) \cdot k}{2k} \cdot \overline{q_x q_\lambda}
        = \frac{k-1}{2k} \cdot \overline{q_x q_\lambda}^2
        \text{.}
    \]

    Since we also have $\overline{q_\lambda q_y} \leq \overline{p_x q_y}$ by choice of $\lambda$, we similarly obtain $\int_{\lambda}^{y_2} \overline{p_x q_t} \dt \geq \frac{k-1}{2k} \cdot \overline{q_\lambda q_y}^2$.
    Putting these together gives a lower bound for the first term of $\mathrm{cost}_{\| \cdot \|}(\gamma_{xy})$.
    That means
    \[
        \int_{x_2}^{y_2} \overline{p_x q_t} \dt
        = \int_{x_2}^{\lambda} \overline{p_x q_t} \dt + \int_{\lambda}^{y_2} \overline{p_x q_t} \dt
        \geq \frac{k-1}{2k} \cdot (\overline{q_x q_\lambda}^2 + \overline{q_\lambda q_y}^2)
        \geq \frac{k-1}{4k} \cdot \overline{q_x q_y}^2
    \]
    by utilising $2 \cdot (\overline{q_x q_\lambda}^2 + \overline{q_\lambda q_y}^2) = (\overline{q_x q_\lambda} + \overline{q_\lambda q_y})^2 + (\overline{q_x q_\lambda} - \overline{q_\lambda q_y})^2 \geq (\overline{q_x q_\lambda} + \overline{q_\lambda q_y})^2 = \overline{q_x q_y}^2$ in the final inequality.
    Proceeding analogously for the second term of $\mathrm{cost}_{\| \cdot \|}(\gamma_{xy})$ results in
    \[
        \mathrm{cost}_{\| \cdot \|}(\gamma_{xy})
        = \int_{x_2}^{y_2} \overline{p_x q_t} \dt + \int_{x_1}^{y_1} \overline{p_s q_y} \ds
        \geq \frac{k-1}{4k} \cdot (\overline{q_x q_y}^2 + \overline{p_x p_y}^2)
        \xrightarrow[]{k \to \infty} \frac{1}{4} \cdot (\overline{p_x p_y}^2 + \overline{q_x q_y}^2)
    \]
    and hence $\overline{p_x p_y}^2 + \overline{q_x q_y}^2 \leq 4 \cdot \mathrm{cost}_{\| \cdot \|}(\gamma_{xy})$ as required above, which completes the proof.
    \qed
\end{proof}

\begin{remark}
    \label{rem:tightness}
    It remains open whether the factor of $5$ is tight in general.
    The greatest lower bound that we know of occurs already in 1D:
    Consider $q_x := 0$, $p_x := \phi$, $q_y := \phi + 1$, and $p_y := 2\phi + 1$, where $\phi := (1+\sqrt{5})/2$ is the golden ratio.
    This yields a factor of $2\phi + 1 \approx 4.236$.
    By analysing the trade-off between the inequalities in the above proof, one can likely improve the upper bound in the 1D setting from $5$ to $2\phi + 1$, but closing the gap for dimension $d \geq 2$ seems difficult.
\end{remark}

\autoref{thm:apx-lemma} for the first time allows CDTW approximations without discretisations of any form.
Our approximation algorithm is based on the fact that at least one of the paths $\gamma_{xy}$ and $\gamma_{yx}$ is a happy path whenever both $x$ and $y$ lie on borders of the cell $C$.
At first glance, it thus might even seem like simply omitting all unhappy paths in the propagation procedure of the exact algorithm from~\cite{BuchiBSW2026} yields a polynomial-time result for CDTW in 2D with approximation factor $\beta \leq 5$.

However, this is not the case because minimum-cost happy paths lack an important property of optimal paths:
One can always choose non-crossing optimal paths (cf.~\cite[Lemma~11]{BuchiNW2025}), whereas minimum-cost happy paths may cross, e.g.\ as depicted in the final cell of \autoref{fig:optimal-path-types}.
This prevents the previously mentioned monotone assignment of quadratic pieces.
To achieve a polynomial~bound, we will propagate some unhappy paths while still keeping the resulting complexity low.

\subsection{Algorithm Outline}
\label{subsec:algorithm-outline}

Having established one of the main building blocks, we now provide a conceptual outline of our approximation algorithm, which uses a propagation scheme tailored to \autoref{def:happy-and-unhappy} and \autoref{thm:apx-lemma}.
Although the outline has no requirements on the norm $\tnorm$, note that it only seems feasible to computationally deal with the propagated cost functions if they are well-behaved, such as under polygonal norms.
We defer the implementation details along with the analysis to \autoref{sec:algorithm-analysis}.

The dynamic program in \autoref{alg:cdtw-apx} visits all parameter space cells in layers, starting at~$C_{1,1}$ and ending at~$C_{n,m}$.
For each cell $C$ we store its north, east, south and west border as $C.\kernedN$, $C.\kernedE$, $C.\kernedS$ and~$C.\kernedW$ respectively.
An output border $\mathcal{B}$ of $C$ knows its \emph{adjoining} input border $\mathcal{B}.\mathrm{adj}$ and its \emph{opposing} input border $\mathcal{B}.\mathrm{opp}$.
That means $C.\kernedN.\mathrm{opp} := C.\kernedS =: C.\kernedE.\mathrm{adj}$, so the first three cells of \autoref{fig:optimal-path-types} show paths between opposing borders, and the final two cells show the adjoining~case.
Symmetrically, we thus have $C.\kernedN.\mathrm{adj} := C.\kernedW =: C.\kernedE.\mathrm{opp}$ as well.
(Cf.~\cite[Definition~13]{BuchiBSW2026}.)

Our algorithm computes for every cell border $\mathcal{B}$ a cost function $\mathcal{B}.\mathrm{apx}$ that $5$-approximates the optimum function $\mathrm{opt}_{0,\mathcal{B}}$.
We further store a value $h_{i,j}$ for all $(i,j) \in \{0,\dotsc,n\} \times \{0,\dotsc,m\}$, which is set to the minimum computed cost from the origin $\mathbf{0}$ to the north-east corner $(\| P_{\leq i} \|, \| Q_{\leq j} \|)^{\mathsf{T}}$ of the cell $C_{i,j}$.
The base case of the dynamic program in lines~1--2 works as usual:
It initialises all costs given by paths that travel either straight horizontally or straight vertically from $\mathbf{0}$.

Then the main loop initialises in lines~6--7 for each cell $C$ the approximate cost functions of its output borders $C.\kernedN$ and $C.\kernedE$ via straight paths from the cell corners.
In the following, the costs of optimal paths within $C$ are implicitly given by some choice of $\ell$ as in \autoref{thm:optimal-paths}.

\begin{algorithm}[H]
    \caption{A $5$-approximation for CDTW of polygonal curves $P,Q$ in $(\mathbb{R}^2,\tnorm)$}
    \label{alg:cdtw-apx}
    \begin{algorithmic}[1]
        \normalsize
        \iFor{$i \gets 1$ \algTo $n$} $C_{i,0}.\kernedN.\mathrm{apx} \gets [t \mapsto \smash{\mathrm{opt}_{\| \cdot \|}}(\mathbf{0},(t,0)^\mathsf{T})]$; $h_{i,0} \gets C_{i,0}.\kernedN.\mathrm{apx}(\| P_{\leq i} \|)$
        \iFor{$j \gets 1$ \algTo $m$} $C_{0,j}.\kernedE.\mathrm{apx} \gets [t \mapsto \mathrm{opt}_{\| \cdot \|}(\mathbf{0},(0,t)^\mathsf{T})]$; $h_{0,j} \gets C_{0,j}.\kernedE.\mathrm{apx}(\| Q_{\leq j} \|)$
        \For{$l \gets 2$ \algTo $n + m$}
            \Foreach{$C \gets C_{i,j}$ \algWith $i + j = l$}
                    \Comment{visit the layer-$l$ cells}
                \State $C.\kernedS.\mathrm{apx} \gets C_{i,j-1}.\kernedN.\mathrm{apx}$; $C.\kernedW.\mathrm{apx} \gets C_{i-1,j}.\kernedE.\mathrm{apx}$
                \State $C.\kernedN.\mathrm{apx} \gets [t \mapsto h_{i-1,j} + \mathrm{opt}_{C.\kernedW[],C.\kernedN[]}(\| Q_{\leq j} \|,t)]$
                    \Comment{propagate corner NW}
                \State $C.\kernedE.\mathrm{apx} \gets [t \mapsto h_{i,j-1} + \mathrm{opt}_{C.\kernedS[],C.\kernedE[]}(\| P_{\leq i} \|,t)]$
                    \Comment{propagate corner SE}
                \State $H_\mathcal{S} \gets \{ C.\kernedS.\mathrm{apx}(s) + \mathrm{opt}_{C.\kernedS[],C.\kernedE[]}(s,\| Q_{\leq j} \|) \mid s \in \mathrm{dom}(C.\kernedS[]) \}$
                \State $H_\mathcal{W} \gets \{ C.\kernedW.\mathrm{apx}(s) + \mathrm{opt}_{C.\kernedW[],C.\kernedN[]}(s,\| P_{\leq i} \|) \mid s \in \mathrm{dom}(C.\kernedW[]) \}$
                \State $C.\propagatecosts(\min H_\mathcal{S}, \arg \min H_\mathcal{S}, \min H_\mathcal{W}, \arg \min H_\mathcal{W})$
                    \Comment{call subroutine}
                \State \smash{$h_{i,j} \gets \min H_\mathcal{S} \cup H_\mathcal{W}$}
                    \Comment{store cost of corner NE}
            \EndForeach
        \EndFor
        \Return{$h_{n,m}$}
    \end{algorithmic}
\end{algorithm}

Next, lines~8--11 of \autoref{alg:cdtw-apx} identify the best path from each of the two input borders to the north-east corner of $C$.
See \autoref{fig:algorithm-example-best-paths} for a possible configuration.
Our propagation scheme utilises these paths' costs and starting points, which are passed to the subroutine~$C.\propagatecosts$.
This is a difference to previous propagation-based CDTW algorithms \cite{Klare2020,BuchiNW2025,Brank2022,BuchiBSW2026}, although a related idea was stated in \cite[Observation~3.3]{BuchiBW2009} within the context of the partial Fréchet similarity.

As outlined in \autoref{alg:cdtw-subroutines}, every input border $\mathcal{A}$ of $C$ is then propagated to its adjoining output border $\mathcal{B}$ using the subroutine $\mathcal{B}.\propagatefromadj$, which updates the cost function $\mathcal{B}.\mathrm{apx}$.
It only propagates paths with the same starting point on $\mathcal{A}$ as its best path to the north-east corner, see \autoref{fig:algorithm-example-adj-1}.
This may include unhappy paths with that starting point, e.g.\ as in \autoref{fig:algorithm-example-adj-2}, but it omits all other unhappy paths in contrast to exact algorithms.

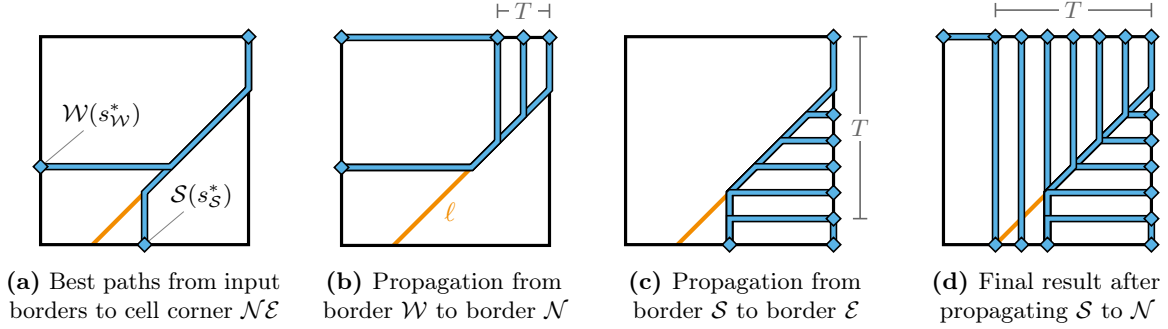
\begin{figure}[H]%
    \centering%
    \begin{subfigure}{0.225\linewidth}
        \centering
        \begin{tikzpicture}[x=2.75cm,y=2.75cm]
            \path[valley] (0.25,0) -- (1,0.75);
            \path[main line] (0,0) rectangle (1,1);

            \path[matching path] (0.5,0) -- (0.5,0.25) -- (1,0.75) -- (1,1);
            \path[matching path] (0,0.375) -- (0.625,0.375) -- (1,0.75) -- (1,1);

            \node[matching point,pin={55:$\mathcal{S}(s^*_{\mathcal{S}})$}] at (0.5,0) {};
            \node[matching point,pin={63:$\mathcal{W}(s^*_{\mathcal{W}})$}] at (0,0.375) {};
            \node[matching point] at (1,1) {};
        \end{tikzpicture}
        \subcaption{Best paths from input borders to cell corner $\mathcal{NE}$}
        \label{fig:algorithm-example-best-paths}
    \end{subfigure}%
    \hspace*{0.025\linewidth}%
    \begin{subfigure}{0.225\linewidth}
        \centering
        \begin{tikzpicture}[x=2.75cm,y=2.75cm]
            \path[valley] (0.25,0) -- node[pos=0.3,below right,spacing] {\color{UmiOrange}$\ell$} (1,0.75);
            \path[main line] (0,0) rectangle (1,1);

            \path[matching path] (0,0.375) -- (0.625,0.375) -- (1,0.75) -- (1,1);
            \path[matching path] (0,0.375) -- (0.625,0.375) -- (0.875,0.625) -- (0.875,1);
            \path[matching path] (0,0.375) -- (0.625,0.375) -- (0.75,0.5) -- (0.75,1);
            \path[matching path] (0,1) -- (0.75,1);

            \node[matching point,draw=none,fill=none] at (0.5,0) {};
            \node[matching point] at (0,0.375) {};
            \node[matching point] at (1,1) {};
            \node[matching point] at (0.875,1) {};
            \node[matching point] at (0.75,1) {};
            \node[matching point] at (0,1) {};

            \path[interval n] (0.75,1) -- node {$T$} (1,1);
        \end{tikzpicture}
        \subcaption{Propagation from border $\mathcal{W}$ to border $\mathcal{N}$}
        \label{fig:algorithm-example-adj-1}
    \end{subfigure}%
    \hspace*{0.025\linewidth}%
    \begin{subfigure}{0.225\linewidth}
        \centering
        \begin{tikzpicture}[x=2.75cm,y=2.75cm]
            \path[valley] (0.25,0) -- (1,0.75);
            \path[main line] (0,0) rectangle (1,1);

            \path[matching path] (0.5,0) -- (0.5,0.25) -- (1,0.75) -- (1,1);
            \path[matching path] (0.5,0) -- (0.5,0.25) -- (0.875,0.625) -- (1,0.625);
            \path[matching path] (0.5,0) -- (0.5,0.25) -- (0.75,0.5) -- (1,0.5);
            \path[matching path] (0.5,0) -- (0.5,0.25) -- (0.625,0.375) -- (1,0.375);
            \path[matching path,line join=bevel] (0.5,0) -- (0.5,0.25) -- (0.5,0.25) -- (1,0.25);
            \path[matching path] (0.5,0) -- (0.5,0.125) -- (1,0.125);
            \path[matching path] (1,0) -- (1,0.125);

            \node[matching point] at (0.5,0) {};
            \node[matching point] at (1,1) {};
            \node[matching point] at (1,0.625) {};
            \node[matching point] at (1,0.5) {};
            \node[matching point] at (1,0.375) {};
            \node[matching point] at (1,0.25) {};
            \node[matching point] at (1,0.125) {};
            \node[matching point] at (1,0) {};

            \path[interval e] (1,0.125) -- node {$T$} (1,1);
        \end{tikzpicture}
        \subcaption{Propagation from border $\mathcal{S}$ to border $\mathcal{E}$}
        \label{fig:algorithm-example-adj-2}
    \end{subfigure}%
    \hspace*{0.025\linewidth}%
    \begin{subfigure}{0.225\linewidth}
        \centering
        \begin{tikzpicture}[x=2.75cm,y=2.75cm]
            \path[valley] (0.25,0) -- (1,0.75);
            \path[main line] (0,0) rectangle (1,1);

            \path[matching path] (0.5,0) -- (0.5,1);
            \path[matching path] (0.5,0) -- (0.5,0.25) -- (1,0.75) -- (1,1);
            \path[matching path] (0.5,0) -- (0.5,0.25) -- (0.875,0.625) -- (1,0.625);
            \path[matching path] (0.5,0) -- (0.5,0.25) -- (0.875,0.625) -- (0.875,1);
            \path[matching path] (0.5,0) -- (0.5,0.25) -- (0.75,0.5) -- (1,0.5);
            \path[matching path] (0.5,0) -- (0.5,0.25) -- (0.75,0.5) -- (0.75,1);
            \path[matching path] (0.5,0) -- (0.5,0.25) -- (0.625,0.375) -- (1,0.375);
            \path[matching path] (0.5,0) -- (0.5,0.25) -- (0.625,0.375) -- (0.625,1);
            \path[matching path,line join=bevel] (0.5,0) -- (0.5,0.25) -- (0.5,0.25) -- (1,0.25);
            \path[matching path] (0.5,0) -- (0.5,0.125) -- (1,0.125);
            \path[matching path] (1,0) -- (1,0.125);
            \path[matching path] (0.375,0) -- (0.375,1);
            \path[matching path] (0.25,0) -- (0.25,1);
            \path[matching path] (0,1) -- (0.25,1);

            \node[matching point] at (0.5,0) {};
            \node[matching point] at (0.5,1) {};
            \node[matching point] at (1,1) {};
            \node[matching point] at (1,0.625) {};
            \node[matching point] at (0.875,1) {};
            \node[matching point] at (1,0.5) {};
            \node[matching point] at (0.75,1) {};
            \node[matching point] at (1,0.375) {};
            \node[matching point] at (0.625,1) {};
            \node[matching point] at (1,0.25) {};
            \node[matching point] at (1,0.125) {};
            \node[matching point] at (1,0) {};
            \node[matching point] at (0.375,0) {};
            \node[matching point] at (0.375,1) {};
            \node[matching point] at (0.25,0) {};
            \node[matching point] at (0.25,1) {};
            \node[matching point] at (0,1) {};

            \path[interval n] (0.25,1) -- node {$T$} (1,1);
        \end{tikzpicture}
        \subcaption{Final result after propagating $\mathcal{S}$ to $\mathcal{N}$}
        \label{fig:algorithm-example-opp}
    \end{subfigure}%
    \caption{Example of algorithm steps in case of $\min H_\mathcal{S} = h^*_\mathcal{S} < h^*_\mathcal{W} = \min H_\mathcal{W}$}%
\end{figure}

In a final step, at most one input border $\mathcal{A}$ is propagated to its opposing output border~$\mathcal{B}$, namely if its best path to the north-east corner is better than the other one.
The subroutine $\mathcal{B}.\propagatefromopp$ updates $\mathcal{B}.\mathrm{apx}$ by computing a lower envelope over all points located on $\mathcal{A}$ up to its best path's starting point.
See \autoref{fig:algorithm-example-opp} for an example result.

\begin{algorithm}[H]
    \caption{Subroutines for propagating costs from input to output borders}
    \label{alg:cdtw-subroutines}
    \begin{algorithmic}[1]
        \normalsize
        \Procedure{$C.$PropagateCosts}{$h^*_\mathcal{S}, s^*_\mathcal{S}, h^*_\mathcal{W}, s^*_\mathcal{W}$}
                \phantomsection
                \label{alg:propagate-costs}
                \Comment{propagate...}
            \State $C.\kernedN.\propagatefromadj(s^*_{\kernedW})$
                \Comment{...W to N}
            \State $C.\kernedE.\propagatefromadj(s^*_\mathcal{S})$
                \Comment{...S to E}
            \iIf{$h^*_\mathcal{S} < h^*_\mathcal{W}$} $C.\kernedN.\propagatefromopp(s^*_\mathcal{S})$
                \Comment{...S to N}
            \iIf{$h^*_\mathcal{W} < h^*_\mathcal{S}$} $C.\kernedE.\propagatefromopp(s^*_\mathcal{W})$
                \Comment{...W to E}
        \EndProcedure

        \vspace*{\doublerulesep}
        \hrule
        \vspace*{\doublerulesep}

        \Procedure{$\mathcal{B}$.PropagateFromAdjoining}{$s^*$}
                \phantomsection
                \label{alg:propagate-from-adj}
            \State $\mathcal{A} \gets \mathcal{B}.\mathrm{adj}$; $\mathrm{apx}_{s^*} \gets [t \mapsto \mathcal{A}.\mathrm{apx}(s^*) + \mathrm{opt}_{\mathcal{A},\mathcal{B}}(s^*,t)]$
                \Comment{use only $s^*$}
            \State $T \gets \{ t \in \mathrm{dom}(\mathcal{B}) \mid \mathrm{apx}_{s^*}(t) < \mathcal{B}.\mathrm{apx}(t) \}$; $\mathcal{B}.\mathrm{apx}|_{T} \gets \mathrm{apx}_{s^*}|_{T}$
                \Comment{update costs}
        \EndProcedure

        \vspace*{\doublerulesep}
        \hrule
        \vspace*{\doublerulesep}

        \Procedure{$\mathcal{B}$.PropagateFromOpposing}{$s^*$}
                \phantomsection
                \label{alg:propagate-from-opp}
            \State $\mathcal{A} \gets \mathcal{B}.\mathrm{opp}$; $\mathrm{apx}_{\leq s^*} \gets [t \mapsto \min_{s \leq s^* \land s \leq t} \mathcal{A}.\mathrm{apx}(s) + \mathrm{opt}_{\mathcal{A},\mathcal{B}}(s,t)]$
                \Comment{use all $s \leq s^*$}
            \State $T \gets \{ t \in \mathrm{dom}(\mathcal{B}) \mid \mathrm{apx}_{\leq s^*}(t) < \mathcal{B}.\mathrm{apx}(t) \}$; $\mathcal{B}.\mathrm{apx}|_{T} \gets \mathrm{apx}_{\leq s^*}|_{T}$
                \Comment{update costs}
        \EndProcedure
    \end{algorithmic}
\end{algorithm}

\section{Algorithm Analysis and Specifics}
\label{sec:algorithm-analysis}

In \autoref{subsec:apx-factor} we show that the above algorithm outline provides a $5$-approximation for CDTW.
Its implementation details and the resulting running time depend on the (polygonal) norm~$\tnorm$.
We first consider $\tnorm[1]$ in \autoref{subsec:running-time} and afterwards extend our results in \autoref{subsec:extension}.

Both the proof of approximation factor and the running time analysis build upon the fact that all approximate cost functions computed by our algorithm are continuous.
To show that, one can apply a useful lemma, which says that converging paths have converging costs.

\begin{lemma}[{\cite[Lemma~19]{BuchiBSW2026}}]
    \label{thm:converging-costs}
    Let $\gamma$ be an arbitrary $(x,y)$-path, and let $(\gamma_k)_{k \in \mathbb{N}}$ be a sequence of $(x_k,y_k)$-paths converging to $\gamma$, which means $\lim_{k \to \infty} (x_k,y_k) = (x,y)$ and $\lim_{k \to \infty} \gamma_k(\lambda) = \gamma(\lambda)$ for all $\lambda \in (\| x \|_1, \| y \|_1)$.
    The limit $\lim_{k \to \infty} \mathrm{cost}_{\| \cdot \|}(\gamma_k)$ exists and is equal to $\mathrm{cost}_{\| \cdot \|}(\gamma)$.
\end{lemma}

\begin{proposition}
    \label{thm:apx-function-properties}
    For every border $\mathcal{B}$ the function $\mathcal{B}.\mathrm{apx}$ computed by our algorithm is continuous.
    Moreover, all $t,t' \in \mathrm{dom}(\mathcal{B})$ with $t < t'$ satisfy $\mathcal{B}.\mathrm{apx}(t) + \mathrm{opt}_{\| \cdot \|}(\mathcal{B}(t), \mathcal{B}(t')) \geq \mathcal{B}.\mathrm{apx}(t')$.
\end{proposition}

The proof of \autoref{thm:apx-function-properties} is deferred to \autoappref{app:apx-function-properties-proof}.
Showing the continuity of $\mathcal{B}.\mathrm{apx}$ through \autoref{thm:converging-costs} is similar to the corresponding proof for the optimum function $\mathrm{opt}_{0,\mathcal{B}}$ \cite[Theorem~20]{BuchiBSW2026}.
The second property from \autoref{thm:apx-function-properties} implies that travelling on a border does not improve the computed costs, which justifies the use of the corner costs $h_{i-1,j}, h_{i,j-1}$ in \autoref{alg:cdtw-apx}.

\subsection{Proof of Approximation Factor}
\label{subsec:apx-factor}

We want to show that the cost of every optimal path from $\mathbf{0}$ to a cell border is $5$-approximated.
To that end, we give a formal and more general definition of best paths to a point.
We further establish the existence of best paths and that they are closed under taking prefix paths.

\begin{definition}
    \label{def:best-path}
    Let $\mathcal{A}$ be an input border of a cell $C$, and let $y \in C$ be some point.
    A \emph{best~path} from $\mathcal{A}$ to $y$ is an optimal $(\mathcal{A}(s^*_y),y)$-path~$\gamma^*_y$ such that $s^*_y \in \mathrm{dom}(\mathcal{A})$ satisfies $\mathcal{A}(s^*_y) \preceq y$ and minimises the function $s \mapsto \mathcal{A}.\mathrm{apx}(s) + \mathrm{opt}_{\| \cdot \|}(\mathcal{A}(s),y)$.
    If the function value at $s^*_y$ is less than the analogous minimum for the other input border, we call $\gamma^*_y$ a \emph{dominating} best path.
\end{definition}

\begin{lemma}
    \label{thm:best-paths}
    For each input border $\mathcal{A}$ of a cell $C$ and each point $y \in C$ there is a best path $\gamma^*_y$ from $\mathcal{A}$ to $y$.
    Now let $\mathcal{A}(s^*_y)$ with $s^*_y \in \mathrm{dom}(\mathcal{A})$ be the starting point of $\gamma^*_y$, and let $\lambda \in [\lambda^*_y, \| y \|_1]$ be arbitrary, where $\lambda^*_y := \| \mathcal{A}(s^*_y) \|_1$.
    Then the prefix path $\gamma^*_y|_{[\lambda^*_y, \lambda]}$ is a best path from $\mathcal{A}$ to $\gamma^*_y(\lambda)$.
    Furthermore, if $\gamma^*_y$ is dominating, we have that $\gamma^*_y|_{[\lambda^*_y, \lambda]}$ is a dominating best path as well.
\end{lemma}

\begin{proof}
    Consider the function $s \mapsto \mathcal{A}.\mathrm{apx}(s) + \mathrm{opt}_{\| \cdot \|}(\mathcal{A}(s),y)$ from \autoref{def:best-path}.
    By \autoref{thm:apx-function-properties}, $\mathcal{A}.\mathrm{apx}$ is continuous.
    \autoref{thm:converging-costs} and \autoref{thm:optimal-paths} imply that $s \mapsto \mathrm{opt}_{\| \cdot \|}(\mathcal{A}(s),y)$ is also continuous.
    Basic calculus says that their sum is a continuous function attaining its minimum on the closed interval $\{ s \in \mathrm{dom}(\mathcal{A}) \mid \mathcal{A}(s) \preceq y \}$ at some $s^*_y$.
    If any prefix path $\gamma^*_y|_{[\lambda^*_y,\lambda]}$ of an optimal $(\mathcal{A}(s^*_y),y)$-path~$\gamma^*_y$ were not a best path, there would be a better path $\gamma_\lambda$ from $\mathcal{A}$ to $\gamma^*_y(\lambda)$.
    Concatenating~$\gamma_\lambda$ with the suffix of $\gamma^*_y$ would yield a better path from $\mathcal{A}$ to $y$ than $\gamma^*_y$, a contradiction.
    Similarly,~the existence of a non-dominating prefix path of $\gamma^*_y$ would imply that $\gamma^*_y$ is non-dominating.
    \qed
\end{proof}

This now enables us to show that for paths whose costs are not propagated by the subroutines from \autoref{alg:cdtw-subroutines} there are propagations of paths with better or approximating costs.
In case of opposing borders, we always find better paths and thus do not require any approximations.

\begin{proposition}
    \label{thm:apx-factor-opp}
    Let $\mathcal{B}$ be an output border of a cell $C$ with opposing input border $\mathcal{A} := \mathcal{B}.\mathrm{opp}$.
    We have $\mathcal{B}.\mathrm{apx}(t) \leq \mathcal{A}.\mathrm{apx}(s) + \mathrm{opt}_{\mathcal{A},\mathcal{B}}(s,t)$ for all $(s,t) \in \mathrm{dom}(\mathcal{A}) \times \mathrm{dom}(\mathcal{B})$ with $s \leq t$.
\end{proposition}

\begin{proof}
    Let $\gamma$ be an optimal $(\mathcal{A}(s),\mathcal{B}(t))$-path.
    We next distinguish between whether the subroutine $\mathcal{B}.\propagatefromopp$ from \autoref{alg:cdtw-subroutines} is called for the given border $\mathcal{B}$ or not.
    Assume first that it is called.
    If $s \leq s^*$ holds, where $s^* := \arg \min H_\mathcal{A}$ is given by lines~8--10 of \autoref{alg:cdtw-apx}, then this call yields $\mathcal{B}.\mathrm{apx}(t) \leq \mathrm{apx}_{\leq s^*}(t) \leq \mathcal{A}.\mathrm{apx}(s) + \mathrm{opt}_{\mathcal{A},\mathcal{B}}(s,t)$.
    Else, we have $s > s^*$ and it follows that $\gamma$ intersects a best path $\gamma^*$ from $\mathcal{A}$ to the north-east corner of $C$ in some point $z \in C$, see \autoref{fig:intersection-of-paths-opp}.
    \autoref{thm:best-paths} says $\mathcal{A}.\mathrm{apx}(s^*) + \mathrm{opt}_{\| \cdot \|}(\mathcal{A}(s^*),z) \leq \mathcal{A}.\mathrm{apx}(s) + \mathrm{opt}_{\| \cdot \|}(\mathcal{A}(s),z)$, so that the call yields $\mathcal{B}.\mathrm{apx}(t) \leq \mathrm{apx}_{\leq s^*}(t) \leq \mathcal{A}.\mathrm{apx}(s^*) + \mathrm{opt}_{\mathcal{A},\mathcal{B}}(s^*,t) \leq \mathcal{A}.\mathrm{apx}(s) + \mathrm{opt}_{\mathcal{A},\mathcal{B}}(s,t)$.

    Assume now that $\mathcal{B}.\propagatefromopp$ is not called.
    Because of the decisions made in lines~4--5 of \autoref{alg:cdtw-subroutines}, there necessarily exists a dominating best path $\gamma^*_0$ from the other input border $\mathcal{A}_0 := \mathcal{B}.\mathrm{adj}$ to the north-east corner of $C$.
    The path $\gamma^*_0$ intersects every path from $\mathcal{A}$ to~$\mathcal{B}$, including $\gamma$.
    Let $z_0 \in C$ be an intersection point of these two paths, and let $\mathcal{A}_0(s^*_0)$ be the starting point of $\gamma^*_0$.
    We have $\mathcal{A}_0.\mathrm{apx}(s^*_0) + \mathrm{opt}_{\| \cdot \|}(\mathcal{A}_0(s^*_0),z_0) < \mathcal{A}.\mathrm{apx}(s) + \mathrm{opt}_{\| \cdot \|}(\mathcal{A}(s),z_0)$ by \autoref{thm:best-paths}.
    The subroutine call $\mathcal{B}.\propagatefromadj$, which happens in lines~2--3 of \autoref{alg:cdtw-subroutines}, now results in $\mathcal{B}.\mathrm{apx}(t) \leq \mathrm{apx}_{s^*_0}(t) = \mathcal{A}_0.\mathrm{apx}(s^*_0) + \mathrm{opt}_{\mathcal{A}_0,\mathcal{B}}(s^*_0,t) < \mathcal{A}.\mathrm{apx}(s) + \mathrm{opt}_{\mathcal{A},\mathcal{B}}(s,t)$.
    \qed
\end{proof}

In case of adjoining borders, we get a similar result for happy paths, whereas unhappy paths might not intersect any propagated best paths.
We can, however, use \autoref{thm:apx-lemma} to show that all unhappy paths are already approximated by the propagation of the related cell corner.

\filbreak

\begin{figure}[H]%
    \centering%
    \begin{subfigure}{0.3\linewidth}
        \centering
        \begin{tikzpicture}[x=3.25cm,y=3.25cm,trim left=0cm,trim right=3.25cm]
            \path[valley] (0,0) -- (1,1);
            \path[main line] (0,0) rectangle (1,1);

            \path[matching path] (1/3,0) -- (1/3,1/3) -- (1,1);
            \node[matching point,label={below:$\mathcal{A}(s^*)$}] at (1/3,0) {};
            \node[matching point] at (1,1) {};

            \path[matching path] (2/3,0) -- (2/3,1);
            \node[matching point,label={below:$\mathcal{A}(s)$}] at (2/3,0) {};
            \node[matching point,label={above:$\mathcal{B}(t)$}] at (2/3,1) {};

            \node[inner sep=2pt,pin={below right:$z$}] at (2/3,2/3) {};
        \end{tikzpicture}
        \subcaption{Intersection with best path \\ in case of opposing borders}
        \label{fig:intersection-of-paths-opp}
    \end{subfigure}%
    \hspace*{0.1\linewidth/3}%
    \begin{subfigure}{0.3\linewidth}
        \centering
        \begin{tikzpicture}[x=3.25cm,y=3.25cm,trim left=0cm,trim right=3.25cm]
            \path[valley] (0,0) -- (1,1);
            \path[main line] (0,0) rectangle (1,1);

            \path[matching path] (0.25,0) -- (0.25,0.25) -- (0.375,0.375) -- (1,0.375);
            \node[matching point,label={below:$\mathcal{A}(s)$}] at (0.25,0) {};
            \node[matching point,pin={120:$\mathcal{B}(t)$}] at (1,0.375) {};

            \path[matching path] (0.625,0) -- (0.625,0.625) -- (1,1);
            \node[matching point,label={below:$\mathcal{A}(s^*)$}] at (0.625,0) {};
            \node[matching point] at (1,1) {};

            \node[inner sep=2.5pt,pin={below left:$z$}] at (0.625,0.375) {};
        \end{tikzpicture}
        \subcaption{Intersection with best path \\ in case of adjoining borders}
        \label{intersection-of-paths-adj}
    \end{subfigure}%
    \hspace*{0.1\linewidth/3}%
    \begin{subfigure}{0.3\linewidth}
        \centering
        \begin{tikzpicture}[x=3.25cm,y=3.25cm,trim left=0cm,trim right=3.25cm]
            \path[valley] (0,1/3) -- node[pos=0.5,above left,spacing] {\color{UmiOrange}$\ell$} (2/3,1);
            \path[main line] (0,0) rectangle (1,1);

            \path[matching path] (0.5,0) |- (1,0.5);
            \path[matching path,double=UmiGreen] (1,0) -- (1,0.5);
            \path[matching path,double=UmiGreen,util dash] (0.5,0) -- (1,0);

            \node[matching point,split fill={UmiSkyblue}{UmiGreen},label={below:$\mathcal{A}(s)$}] (s) at (0.5,0) {};
            \node[matching point,split fill={UmiSkyblue}{UmiGreen},pin={120:$\mathcal{B}(t)$}] at (1,0.5) {};
            \node[matching point,fill=UmiGreen,label={below:$c\vphantom{\mathcal{A}(s)}$}] (c) at (1,0) {};

            \begin{scope}[on background layer,shift={(0.5,0)},scale=0.5]
                \apxhatch{5}
            \end{scope}
        \end{tikzpicture}
        \subcaption{Approximation of unhappy path by cell corner path}
        \label{fig:approximation-of-path}
    \end{subfigure}%
    \caption{Different configurations yielding a path with at worst approximating cost}%
\end{figure}
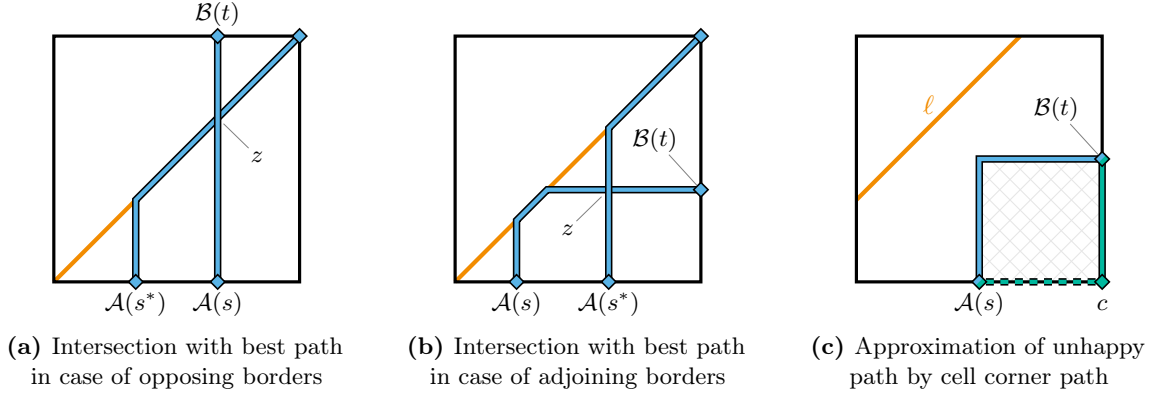

\begin{proposition}
    \label{thm:apx-factor-adj}
    Let $\mathcal{B}$ be an output border of a cell $C$ with adjoining input border $\mathcal{A} := \mathcal{B}.\mathrm{adj}$.
    We have $\mathcal{B}.\mathrm{apx}(t) \leq \mathcal{A}.\mathrm{apx}(s) + 5 \cdot \mathrm{opt}_{\mathcal{A},\mathcal{B}}(s,t)$ for all $(s,t) \in \mathrm{dom}(\mathcal{A}) \times \mathrm{dom}(\mathcal{B})$.
\end{proposition}

\begin{proof}
    Let $\gamma$ be an optimal $(\mathcal{A}(s),\mathcal{B}(t))$-path as in \autoref{thm:optimal-paths}.
    Assume first that $\gamma$ is of the form from \autoref{thm:optimal-paths}a, which implies that $\gamma$ intersects every path from $\mathcal{A}$ via $\ell$ to the north-east corner of $C$, either on $\ell$ or else as in \autoref{intersection-of-paths-adj}.
    This includes a best path $\gamma^*$ with starting point $\mathcal{A}(s^*)$, where $s^* := \arg \min H_\mathcal{A}$ is given by lines~8--10 of \autoref{alg:cdtw-apx}.
    Let $z \in C$ be an intersection point of $\gamma$ and $\gamma^*$.
    \autoref{thm:best-paths} says $\mathcal{A}.\mathrm{apx}(s^*) + \mathrm{opt}_{\| \cdot \|}(\mathcal{A}(s^*),z) \leq \mathcal{A}.\mathrm{apx}(s) + \mathrm{opt}_{\| \cdot \|}(\mathcal{A}(s),z)$, so that the subroutine call $\mathcal{B}.\propagatefromadj$, which happens in lines~2--3 of \autoref{alg:cdtw-subroutines}, results in $\mathcal{B}.\mathrm{apx}(t) \leq \mathrm{apx}_{s^*}(t) = \mathcal{A}.\mathrm{apx}(s^*) + \mathrm{opt}_{\mathcal{A},\mathcal{B}}(s^*,t) \leq \mathcal{A}.\mathrm{apx}(s) + \mathrm{opt}_{\mathcal{A},\mathcal{B}}(s,t)$.

    Assume now that $\gamma$ is of the form from \autoref{thm:optimal-paths}b.
    By \autoref{thm:apx-lemma}, we then have $\mathcal{A}.\mathrm{apx}(s) + \allowbreak 5 \cdot \mathrm{opt}_{\mathcal{A},\mathcal{B}}(s,t) = \mathcal{A}.\mathrm{apx}(s) + 5 \cdot \mathrm{cost}_{\| \cdot \|}(\gamma) \geq \mathcal{A}.\mathrm{apx}(s) + \mathrm{cost}_{\| \cdot \|}(\gamma_c)$, where $\gamma_c$ is the $(\mathcal{A}(s),\mathcal{B}(t))$-path through the shared cell corner $c = \mathcal{A}(s_c)$ of the adjoining borders $\mathcal{A}$ and $\mathcal{B}$, see \autoref{fig:approximation-of-path}.
    We obtain $\mathcal{A}.\mathrm{apx}(s) + \mathrm{cost}_{\| \cdot \|}(\gamma_c) \geq \mathcal{A}.\mathrm{apx}(s_c) + \mathrm{opt}_{\| \cdot \|}(c, \mathcal{B}(t)) \geq \mathcal{B}.\mathrm{apx}(t)$ due to \autoref{thm:apx-function-properties} and the initialisation of $\mathcal{B}.\mathrm{apx}$ in lines~6--7 of \autoref{alg:cdtw-apx}, which completes the proof.
    \qed
\end{proof}

We conclude that the conceptual algorithm outline from \autoref{subsec:algorithm-outline} indeed gives a CDTW approximation within factor $\beta \leq 5$.
Using the above propositions, this follows by induction.

\begin{theorem}
    \label{thm:apx-factor-algo}
    Any implementation of \autoref{alg:cdtw-apx} and \autoref{alg:cdtw-subroutines} under a norm $\tnorm$ returns a factor-$5$ approximation of $\mathrm{cdtw}_{\| \cdot \|}(P,Q)$ for all polygonal curves $P,Q$ in $(\mathbb{R}^2,\tnorm)$.
\end{theorem}

\begin{proof}
    We claim that for every border $\mathcal{B}$ the computed cost function $\mathcal{B}.\mathrm{apx}$ is a $5$-approximation of the optimum function $\mathrm{opt}_{0,\mathcal{B}}$.
    Clearly, the lower bound $\mathcal{B}.\mathrm{apx}(t) \geq \mathrm{opt}_{0,\mathcal{B}}(t)$ is satisfied for all $t \in \mathrm{dom}(\mathcal{B})$ since $\mathcal{B}.\mathrm{apx}(t)$ is equal to the cost of a $(\mathbf{0},\mathcal{B}(t))$-path by \autoref{alg:cdtw-apx} and \autoref{alg:cdtw-subroutines}.
    In the following, we show the upper bound $\mathcal{B}.\mathrm{apx}(t) \leq 5 \cdot \mathrm{opt}_{0,\mathcal{B}}(t)$ by induction.
    This gives
    \[
        \mathrm{cdtw}_{\| \cdot \|}(P,Q)
        = \mathrm{opt}_{\| \cdot \|}(\mathbf{0},(\| P \|, \| Q \|)^{\mathsf{T}})
        \leq h_{n,m}
        \leq 5 \cdot \mathrm{opt}_{\| \cdot \|}(\mathbf{0},(\| P \|, \| Q \|)^{\mathsf{T}})
        = 5 \cdot \mathrm{cdtw}_{\| \cdot \|}(P,Q)
    \]
    as desired, where $h_{n,m} = C_{n,m}.\kernedN.\mathrm{apx}(\| P \|) = C_{n,m}.\kernedE.\mathrm{apx}(\| Q \|)$ is returned by \autoref{alg:cdtw-apx}.

    The induction's base case is given by lines 1--2 of \autoref{alg:cdtw-apx}:
    For each initialised border~$\mathcal{B}$ and each $t \in \mathrm{dom}(\mathcal{B})$ there exists a single $(\mathbf{0},\mathcal{B}(t))$-path, so $\mathcal{B}.\mathrm{apx}(t) = \mathrm{opt}_{0,\mathcal{B}}(t) \leq 5 \cdot \mathrm{opt}_{0,\mathcal{B}}(t)$.
    As for the inductive step:
    Let $\mathcal{B}$ be an output border of a cell $C$, and let $t \in \mathrm{dom}(\mathcal{B})$ be arbitrary.
    Because every $(\mathbf{0},\mathcal{B}(t))$-path must travel across one of the input borders of $C$, we have $\mathrm{opt}_{0,\mathcal{B}}(t) = \allowbreak \mathrm{opt}_{0,\mathcal{A}}(s) + \mathrm{opt}_{\mathcal{A},\mathcal{B}}(s,t)$ for an input border $\mathcal{A}$ and an $s \in \mathrm{dom}(\mathcal{A})$ with $\mathcal{A}(s) \preceq \mathcal{B}(t)$, cf.~\autoref{thm:best-paths}.
    By \autoref{thm:apx-factor-opp} or \autoref{thm:apx-factor-adj} as well as the induction hypothesis, we therefore obtain
    \[
        \mathcal{B}.\mathrm{apx}(t)
        \leq \mathcal{A}.\mathrm{apx}(s) + 5 \cdot \mathrm{opt}_{\mathcal{A},\mathcal{B}}(s,t)
        \leq 5 \cdot \mathrm{opt}_{0,\mathcal{A}}(s) + 5 \cdot \mathrm{opt}_{\mathcal{A},\mathcal{B}}(s,t)
        = 5 \cdot \mathrm{opt}_{0,\mathcal{B}}(t)
        \text{,}
    \]
    which is the upper bound claimed above.
    This completes the proof of approximation factor.
    \qed
\end{proof}

As discussed in \autoref{rem:tightness}, we do not know whether the factor of $5$ is tight.
Even the value $\beta$ that is tight for \autoref{thm:apx-lemma} might not be tight for \autoref{thm:apx-factor-algo}.
It would only transfer if there are curves $P,Q$ with an optimal $(\mathbf{0},(\| P \|, \| Q \|)^{\mathsf{T}})$-path that is predominantly unhappy.

\subsection{Running Time under the 1-Norm}
\label{subsec:running-time}

In this section we specify how to implement our approximation algorithm under the $1$-norm~$\tnorm[1]$, and analyse its running time on two curves $P,Q$ of complexity $n$ and $m$ respectively.
By~definition, we have $\| z \|_1 = |z_1| + |z_2|$ for $z \in \mathbb{R}^2$, meaning the $1$-norm is linear on each quadrant of the~plane.
Every cell border $\mathcal{B}$ hence contains only $O(1)$ breakpoints where the function $t \mapsto \| P_{\| \cdot \|_1}(\mathcal{B}_1(t)) - \allowbreak Q_{\| \cdot \|_1}(\mathcal{B}_2(t)) \|_1$ switches from one linear function piece to another.
It follows that the base case in lines 1--2 of \autoref{alg:cdtw-apx}, which uses straight paths, has piecewise linear integrands and yields~$O(1)$ quadratic pieces on each initialised border $\mathcal{B}$.
Computing all pieces requires $O(n + m)$ time.

Similarly, propagating the costs of straight paths from the cell corners in lines~6--7 yields~$O(1)$ initial quadratic pieces per output border of the cell $C$.
We now claim that lines~8--11 of~\autoref{alg:cdtw-apx} take $O(N_{C.\kernedS} + N_{C.\kernedW})$ time, where $N_\mathcal{B}$ denotes the number of (maximal) quadratic pieces making up the computed cost function $\mathcal{B}.\mathrm{apx}$ of a border $\mathcal{B}$.
Note that the costs of optimal paths within~$C$ are given by \autoref{thm:optimal-paths}, and a suitable $\ell$ is computable in $O(1)$ time under $\tnorm[1]$ \cite[Corollary~9]{BuchiBSW2026}.
First, we consider the possible starting points of best paths to the north-east corner of $C$.

\begin{definition}
    \label{def:parent}
    Let $\mathcal{A},\mathcal{B}$ be a pair of input and output border, and let $(s_0,t_0) \in \mathrm{dom}(\mathcal{A}) \times \mathrm{dom}(\mathcal{B})$ with $\mathcal{A}(s_0) \preceq \mathcal{B}(t_0)$ be arbitrary.
    If the function $s \mapsto \mathcal{A}.\mathrm{apx}(s) + \mathrm{opt}_{\mathcal{A},\mathcal{B}}(s,t_0)$ attains a semistrict\,%
    \footnote{That means strict to the left and/or to the right, which rules out interior points of constant-valued pieces.} %
    local minimum at $s_0$, we say that $\mathcal{A}(s_0)$ is a \emph{parent candidate} for $\mathcal{B}(t_0)$.
    If $\mathcal{B}.\mathrm{apx}(t_0) = \mathcal{A}.\mathrm{apx}(s_0) + \mathrm{opt}_{\mathcal{A},\mathcal{B}}(s_0,t_0)$ holds, we say that $\mathcal{A}(s_0)$ is a \emph{parent point} of $\mathcal{B}(t_0)$.
    We call an interval $S \subseteq \mathrm{dom}(\mathcal{A})$ a \emph{parent interval} of $T \subseteq \mathrm{dom}(\mathcal{B})$ if each $\mathcal{B}(t)$ with $t \in T$ has a parent point $\mathcal{A}(s)$ with $s \in S$.
\end{definition}

This definition builds upon \cite[Definition~18]{BuchiNW2025}, though it allows for a more fine-grained analysis:
Whereas \cite{BuchiNW2025} collects consecutive quadratic pieces into specific \emph{subsegments} and then introduces a parent relationship for these subsegments, we deal with arbitrary points and intervals.

For a fixed $t_0$, as in lines~8--9 of \autoref{alg:cdtw-apx}, the given function $s \mapsto \mathcal{A}.\mathrm{apx}(s) + \mathrm{opt}_{\mathcal{A},\mathcal{B}}(s,t_0)$ is piecewise quadratic with $N_\mathcal{A} + O(1)$ pieces under $\tnorm[1]$.
This is because $\mathcal{A}.\mathrm{apx}$ has $N_\mathcal{A}$~pieces, and the optimal $(\mathcal{A}(s),\mathcal{B}(t_0))$-paths from \autoref{thm:optimal-paths} have $O(1)$ bending points whose coordinates may depend linearly on the variable $s$.
Hence, $s \mapsto \mathrm{opt}_{\mathcal{A},\mathcal{B}}(s,t_0)$ has $O(1)$ quadratic pieces by the linearity of $\tnorm[1]$ on each quadrant.
Basic calculus says that (semistrict) local minima of~a~piece\-wise quadratic function can only occur at domain endpoints, at breakpoints between~pieces, or at points where the derivative is $0$.
The latter correspond to parabola vertices, see \autoref{fig:local-minima}.

\begin{figure}[H]
    \centering
    \begin{tikzpicture}[x=3cm,y=3cm]
        \useasboundingbox (0.025,1.095) rectangle (4.475,0.0325);
        \path[axes] (0.025,1.095) |- (4.475,0.0325);

        \path[main line,UmiBlue] (0.2,0.32) parabola bend (0.5,0.5) (0.9,0.18);
        \node[round point,fill=UmiYellow] at (0.2,0.32) {};
        \node[round point,fill=white] at (0.5,0.5) {};

        \begin{scope}
            \clip (0.9,0.18) rectangle (2.2,0.87);
            \path[main line,UmiBlue] (0.9,0.18)
                parabola bend (0.8,0.17) (1.5,0.66)
                parabola bend (1.8,0.39) (2.2,0.87);
        \end{scope}

        \node[round point,fill=UmiYellow] at (0.9,0.18) {};
        \node[round point] at (1.5,0.66) {};
        \node[round point,fill=UmiYellow] at (1.8,0.39) {};

        \path[main line,UmiBlue] (2.2,0.87)
            parabola bend (2.5,0.96) (3.05,0.6575)
            parabola bend (3.4,0.1675) (3.7,0.5275)
            parabola[bend at start] (4.3,0.3475);

        \node[round point] at (2.2,0.87) {};
        \node[round point,fill=white] at (2.5,0.96) {};
        \node[round point] at (3.05,0.6575) {};
        \node[round point,fill=UmiYellow] at (3.4,0.1675) {};
        \node[round point] at (3.7,0.5275) {};
        \node[round point,fill=UmiYellow] at (4.3,0.3475) {};
    \end{tikzpicture}
    \caption{Local minima of a continuous and piecewise quadratic function}
    \label{fig:local-minima}
\end{figure}
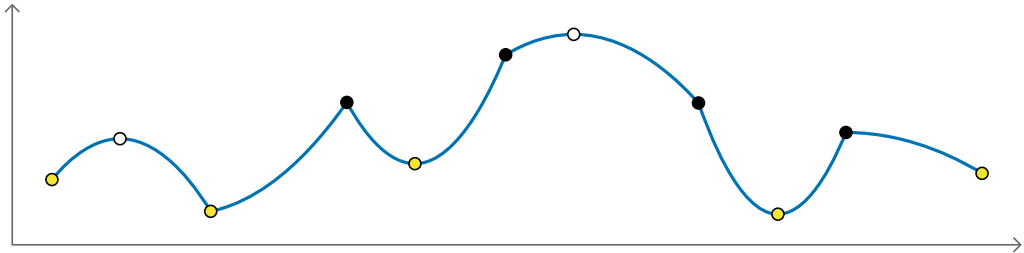

Overall, there are $O(N_{C.\kernedS} + N_{C.\kernedW})$ parent candidates for the north-east corner of $C$.
We~can check all of them to find starting points of best paths, as the global minima that exist by~\autoref{thm:best-paths} are among the local minima.
Thus, computing the required values from lines~10--11 of \autoref{alg:cdtw-apx} takes $O(N_{C.\kernedS} + N_{C.\kernedW})$ time under $\tnorm[1]$.
Moreover, we can reuse this candidate set for opposing propagation since parent points that do not share a coordinate with the child can be chosen from the same set.
This is due to the following key lemma that pertains to happy optimal paths.

\begin{lemma}[Happiness Lemma]
    \label{thm:happiness-lemma}
    Let $X := \{ (s,t) \in S \times T \mid \mathcal{A}(s) \preceq \mathcal{B}(t) \}$, where $\mathcal{A},\mathcal{B}$ is a pair of input and output border with $S \times T \subseteq \mathrm{dom}(\mathcal{A}) \times \mathrm{dom}(\mathcal{B})$.
    If all $(s,t) \in X$ have a happy optimal $(\mathcal{A}(s),\mathcal{B}(t))$-path, then it is $\mathrm{opt}_{\mathcal{A},\mathcal{B}}(s,t) = \rhoinput(s) + \rhooutput(t)$ on $X$ for two univariate functions~$\rhoinput,\rhooutput$.
    In particular, this always applies in case of $\mathcal{A} = \mathcal{B}.\mathrm{opp}$ and $S \times T = \mathrm{dom}(\mathcal{A}) \times \mathrm{dom}(\mathcal{B})$.
\end{lemma}

\begin{proof}
    The function $\mathrm{opt}_{\mathcal{A},\mathcal{B}}$ is obtained by evaluating costs of optimal $(\mathcal{A}(s),\mathcal{B}(t))$-paths in terms of the parameter pair $(s,t) \in X$, and we assume all paths given by \autoref{thm:optimal-paths} to be happy on~$X$.
    Consequently, each bending point $\xi_0$ lies on~$\ell$ or on the cell's boundary according to \autoref{def:happy-and-unhappy}.
    This implies that $\xi_0$ can be expressed in terms of only $s$ or in terms of only $t$, cf.~\autoref{fig:optimal-path-types}.

    We evaluate optimal path costs by summing the costs of subpaths, each of which traces a line segment between two points $\xi$ and $\xi'$ that are among $\mathcal{A}(s)$, $\mathcal{B}(t)$ and the paths' bending points.
    \autoref{def:path-and-cost} says that every such subpath cost is a definite integral of some function $\tnorm \circ \psi$ on the interval $[\| \xi \|_1, \| \xi' \|_1]$, where $\psi$ is an affine map.
    We distinguish the following two cases.

    \begin{enumerate}
        \item
        If the traced line segment lies on $\ell$ or on the cell's boundary, the corresponding function~$\psi$ is fixed and does not depend on $s$ and $t$.
        By the Fundamental Theorem of Calculus, the subpath cost is equal to $\Psi(\| \xi' \|_1) - \Psi(\| \xi \|_1)$, where $\Psi$ is an arbitrary antiderivative of the continuous integrand $\tnorm \circ \psi$.
        Due to happiness, the term $\Psi(\| \xi \|_1)$ may depend only on $s$ and can thus be assigned to $\rhoinput$, while $\Psi(\| \xi' \|_1)$ may depend only on $t$ and can be assigned to $\rhooutput$.

        \item
        Otherwise, the traced segment is horizontal/vertical with $\xi = \mathcal{A}(s)$ or $\xi' = \mathcal{B}(t)$ such that $\psi$ depends only on $s$ or only on $t$.
        Then the other endpoint must depend on the same parameter, i.e.\ $\xi'$ depends on $s$ or $\xi$ depends on $t$.
        Due to happiness, it can be expressed in terms of only this parameter.
        The subpath cost can hence be assigned to $\rhoinput$ or $\rhooutput$ respectively.
    \end{enumerate}

    It remains to show that every $(x,y)$-path $\gamma$ for $x,y$ on opposing borders is happy.
    Since paths from \autoref{thm:optimal-paths}a are always happy, we assume that $\gamma$ is as in \autoref{thm:optimal-paths}b and consider the single breaking point $\xi$ of $\gamma$.
    By construction, $\xi$ shares one coordinate value with~$x$ and another with~$y$.
    Opposing borders parametrise via the same coordinate and keep the other fixed, cf.~\autoref{def:parameter-space-and-borders}.
    It follows that the value of $\xi$ in the non-parametric coordinate is equal to one border's fixed~value.
    This means that $\xi$ lies on the related border and thus on the boundary of $C$, so $\gamma$ is happy.
    \qed
\end{proof}

Being able to split costs into two univariate functions simplifies the behaviour of local minima in terms of one variable, which yields the desired characterisation of parent candidates.

\begin{proposition}
    \label{thm:parent-candidates}
    Let $\mathcal{A}$ be an input border of a cell $C$, and let $S^*_{\mathcal{A}}$ be defined as follows:
    \[
        S^*_{\mathcal{A}} := \{ s \in \mathrm{dom}(\mathcal{A}) \mid \mathcal{A}(s) \text{ is a parent candidate for the north-east corner of } C \}
        \text{.}
    \]
    Further let $t_0 \in \mathrm{dom}(\mathcal{B})$, where $\mathcal{B}$ is the opposing output border of $\mathcal{A}$.
    For every $s \in S^*_\mathcal{A}$ with $s \leq t_0$ we have that $\mathcal{A}(s)$ is a parent candidate for $\mathcal{B}(t_0)$ as well.
    Conversely, each parent candidate~$\mathcal{A}(s_0)$ for~$\mathcal{B}(t_0)$ necessarily satisfies $s_0 \in \{ s \in S^*_\mathcal{A} \mid s \leq t_0 \} \cup \{ t_0 \}$.
    It is $|S^*_\mathcal{A}| \in O(N_\mathcal{A})$ under $\tnorm[1]$.
\end{proposition}

\begin{proof}
    We start with the bound for $|S^*_\mathcal{A}|$ under $\tnorm[1]$.
    As argued above, for a fixed $t_0 \in \mathrm{dom}(\mathcal{B})$ the function $s \mapsto \mathcal{A}.\mathrm{apx}(s) + \mathrm{opt}_{\mathcal{A},\mathcal{B}}(s,t_0)$ from \autoref{def:parent} is piecewise quadratic with $N_\mathcal{A} + O(1)$ pieces under $\tnorm[1]$.
    Apart from domain endpoints and piece breakpoints, each piece can contribute at most one semistrict local minimum with derivative $0$, so we indeed have $|S^*_\mathcal{A}| \in O(N_\mathcal{A})$.

    In general, we have \smash{$\mathcal{A}.\mathrm{apx}(s) + \mathrm{opt}_{\mathcal{A},\mathcal{B}}(s,t_0) = \mathcal{A}.\mathrm{apx}(s) + \rhoinput(s) + \rhooutput(t_0)$} for all $s \in \mathrm{dom}(\mathcal{A})$ with $s \leq t_0$, where $\rhoinput,\rhooutput$ are as in \autoref{thm:happiness-lemma}.
    Because $\rhooutput(t_0)$ does not depend on the variable $s$, all local minima that the function attains on the interior of its domain are entirely determined by $\mathcal{A}.\mathrm{apx}$ and $\rhoinput$.
    The same applies to the fixed left endpoint of the domain, while the right~endpoint is exactly $t_0$ per the constraint $s \leq t_0$ of monotone paths.
    Hence, it is the only parent candidate that can vary.
    The north-east corner of $C$ gives the largest domain, so the stated results follow.
    \qed
\end{proof}

The intuition behind $s_0 \in S^*_\mathcal{A}$ versus $s_0 = t_0$ for opposing propagation is that the corresponding paths may switch between the two types depicted in the second cell of \autoref{fig:optimal-path-types}:
Either they are straight with $s_0 = t_0$, or they travel on $\ell$ after starting at a point $\mathcal{A}(s_0)$ with $s_0 \in S^*_\mathcal{A}$.
In~fact, this switching is a major reason behind the $O((n+m)^5)$ time complexity of the exact 1D algorithm~\cite{BuchiNW2025}.
Note that in 1D the CDTW integrand is $0$ everywhere on $\ell$.
The procedure from \cite{BuchiNW2025} propagates from $\mathcal{A}$ to $\ell$, performs a cumulative minimum operation on $\ell$, and then propagates from $\ell$ to $\mathcal{B}$.
In 2D a cell $C$ is generally associated with two non-parallel curve segments, so that the~CDTW integrand is non-zero within $C$ except at the point that corresponds to their intersection.

\filbreak

Although the cumulative minimum operation can be adapted to non-zero weightings (cf.~\cite[Observation~46]{BuchiNW2025}), \autoref{thm:happiness-lemma} reveals the underlying principles and generalises them to all happy~paths.
This~technical contribution of ours advances the understanding of function propagations.
It allows us to streamline and enhance the running time analysis compared to \cite{BuchiNW2025}, yet it does not hold for unhappy paths.
Their propagation complexity is still unknown in 2D, see \cite[Section~4.2]{BuchiBSW2026}.

We now continue with implementing the subroutines from \autoref{alg:cdtw-subroutines} under the $1$-norm~$\tnorm[1]$.
In $\mathcal{B}.\propagatefromadj$ the cost function $\mathcal{B}.\mathrm{apx}$ gets updated through taking its lower envelope with the function $\mathrm{apx}_{s^*}$ defined by $t \mapsto \mathcal{A}.\mathrm{apx}(s^*) + \mathrm{opt}_{\mathcal{A},\mathcal{B}}(s^*,t)$ on $\mathrm{dom}(\mathcal{B})$.
The given value $s^*$ provides a fixed starting point for propagated paths, which may include unhappy paths.
As fixing the first argument of $\mathrm{opt}_{\mathcal{A},\mathcal{B}}$ is symmetric to fixing its second argument, we have that the function $\mathrm{apx}_{s^*}$ is piecewise quadratic with $O(1)$ pieces under $\tnorm[1]$.
Moreover, $\mathcal{B}.\mathrm{apx}$ has~$O(1)$ pieces before the subroutine calls in lines 2--3 of \autoref{alg:cdtw-subroutines} occur, so taking the lower envelope requires finding only a constant number of intersections between pieces as in \autoref{fig:intersection-of-quadratics}.

\begin{figure}[H]%
    \centering%
    \begin{minipage}[b]{0.45\linewidth}%
        \centering%
        \begin{tikzpicture}[x=2.25cm,y=2.25cm]
            \useasboundingbox (-0.075,1.575) rectangle (3.075,-0.075);
            \path[axes] (-0.075,1.575) |- (3.075,-0.075);

            \begin{scope}
                \clip (0.15,1.5) rectangle (2.85,0);
                \path[name path=l] (0.15,0) -- (0.15,1.5);
                \path[name path=r] (2.85,0) -- (2.85,1.5);

                \path[name path=a] (-2,2.75) parabola bend (0,0.75) (2,2.75);
                \path[name path=b] (-0.35,6.25) parabola bend (2.15,0) (4.65,6.25);
                \path[name path=c] (-0.5,-1.25) parabola bend (3.5,0.75) (7.5,-1.25);

                \path[name intersections={of=l and a,name=la}];
                \path[name intersections={of=a and b,name=ab}];
                \path[name intersections={of=b and c,name=bc}];
                \path[name intersections={of=c and r,name=cr}];

                \path[main line,UmiBlue] (la-1) parabola bend (0,0.75) (ab-1);
                \path[main line,UmiBlue] (bc-1) parabola bend (3.5,0.75) (cr-1);

                \begin{scope}
                    \clip (ab-1) rectangle (bc-1);
                    \path[main line,UmiBlue] (ab-1) parabola bend (2.15,0) (bc-1);
                \end{scope}

                \path[name path=e] (-6.75,6.1) parabola bend (-0.75,0.1) (5.25,6.1);
                \path[name path=f] (-1.75,5) parabola bend (2.25,1) (6.25,5);

                \path[name intersections={of=l and e,name=le}];
                \path[name intersections={of=e and b,name=eb}];
                \path[name intersections={of=e and f,name=ef}];
                \path[name intersections={of=f and r,name=fr}];

                \path[main line,UmiVermillion] (le-1)
                    parabola bend (-0.75,0.1) (ef-1)
                    parabola bend (2.25,1) (fr-1);
            \end{scope}

            \node[round point] at (la-1) {};
            \node[round point] at (ab-1) {};
            \node[round point] at (bc-1) {};
            \node[round point] at (cr-1) {};

            \node[round point] at (le-1) {};
            \node[round point,fill=UmiYellow] at (eb-1) {};
            \node[round point] at (ef-1) {};
            \node[round point] at (fr-1) {};
        \end{tikzpicture}%
        \captionof{figure}{Intersection between \\ piecewise quadratic functions}%
        \label{fig:intersection-of-quadratics}%
    \end{minipage}%
    \hspace*{0.05\linewidth}%
    \begin{minipage}[b]{0.45\linewidth}%
        \centering%
        \begin{tikzpicture}[x=6cm,y=3cm]
            \coordinate (sw) at (-0.075,0);
            \coordinate (ne) at (0.975,1);
            \useasboundingbox ($(sw) + (0,-10pt)$) rectangle ($(ne) + (0,12.5pt)$);

            \path[valley] (0.15,0) -- (0.675,1);
            \path[main line] (sw) rectangle (ne);

            \path[matching path] (-0.075,1) -- (0,1);
            \path[matching path] (0,0) -- (0,1);
            \path[matching path] (0.075,0) -- (0.075,1);
            \path[matching path] (0.15,0) -- (0.15,1);
            \path[matching path] (0.375,0) -- (0.375,1);
            \path[matching path] (0.45,0) -- (0.45,1);
            \path[matching path] (0.525,0) -- (0.525,1);

            \path[matching path] (0.15,0) -- (0.375,3/7) -- (0.375,1);
            \path[matching path] (0.15,0) -- (0.3,2/7) -- (0.3,1);
            \path[matching path] (0.15,0) -- (0.225,1/7) -- (0.225,1);
            \path[matching path] (0.525,0) -- (0.525,5/7) -- (0.675,1) -- (0.825,1);
            \path[matching path] (0.525,0) -- (0.525,5/7) -- (0.6,6/7) -- (0.6,1);
            \path[matching path] (0.825,0) -- (0.825,1) -- (0.975,1);

            \node[matching point] at (0,0) {};
            \node[matching point] at (0.075,0) {};
            \node[point,fill=UmiYellow] at (0.15,0) {};
            \node[point,fill=white] at (0.3,0) {};
            \node[matching point] at (0.375,0) {};
            \node[matching point] at (0.45,0) {};
            \node[point,fill=UmiYellow] at (0.525,0) {};
            \node[point,fill=white] at (0.675,0) {};
            \node[point,fill=UmiYellow,pin={120:$\mathcal{A}(s^*)$}] at (0.825,0) {};

            \path[interval s,|-] (-0.075,0) -- node {$S_1$} (0.15,0);
            \path[interval s,|-] (0.15,0) -- node {$S_2$} (0.3,0);
            \path[interval s,|-] (0.3,0) -- node {$S_3$} (0.525,0);
            \path[interval s,|-] (0.525,0) -- node {$S_4$} (0.675,0);
            \path[interval s] (0.675,0) -- node {$S_5$} (0.825,0);

            \node[matching point] at (-0.075,1) {};
            \node[point,fill=UmiYellow] at (0,1) {};
            \node[matching point] at (0.075,1) {};
            \node[matching point] at (0.15,1) {};
            \node[matching point] at (0.225,1) {};
            \node[matching point] at (0.3,1) {};
            \node[point,fill=UmiYellow] at (0.375,1) {};
            \node[matching point] at (0.45,1) {};
            \node[matching point] at (0.525,1) {};
            \node[matching point] at (0.6,1) {};
            \node[matching point] at (0.675,1) {};
            \node[point,fill=UmiYellow] at (0.825,1) {};
            \node[matching point] at (0.975,1) {};

            \path[interval n,|-] (0,1) -- node {$T_1$} (0.375,1);
            \path[interval n,|-] (0.375,1) -- node {$T_3$} (0.825,1);
            \path[interval n] (0.825,1) -- node {$T_5$} (0.975,1);
        \end{tikzpicture}%
        \captionof{figure}{Monotone assignment of \\ intervals on $\mathcal{B}$ to parents on $\mathcal{A}$}%
        \label{fig:monotone-assignment}%
    \end{minipage}%
\end{figure}

In $\mathcal{B}.\propagatefromopp$ the cost function $\mathcal{B}.\mathrm{apx}$ then gets updated through taking its lower envelope with the function $\mathrm{apx}_{\leq s^*}$ defined by $t \mapsto \min_{s \leq s^* \land s \leq t} \mathcal{A}.\mathrm{apx}(s) + \mathrm{opt}_{\mathcal{A},\mathcal{B}}(s,t)$ on $\mathrm{dom}(\mathcal{B})$.
The function $\mathrm{apx}_{\leq s^*}$ is itself a lower envelope, and \autoref{thm:parent-candidates} implies that it is covered by optimal paths starting at the parent candidates from $\{ s \in S^*_\mathcal{A} \mid s \leq s^* \}$ together with straight paths starting before $\mathcal{A}(s^*)$.
The costs of the former paths are given by $O(N_\mathcal{A})$ functions under~$\tnorm[1]$, each of which again has $O(1)$ quadratic pieces.
The latter paths' costs are given by the function $t \mapsto \mathcal{A}.\mathrm{apx}(t) + \mathrm{opt}_{\mathcal{A},\mathcal{B}}(t,t)$ for $t \leq s^*$, which has at most $N_\mathcal{A} + O(1)$ pieces.

The lower envelope of these pieces with the existing $O(1)$ pieces on $\mathcal{B}$ is computable in linear time since the new pieces are given monotonically by parent points.
Those induce intervals on~$\mathcal{B}$ that have parent intervals on $\mathcal{A}$ separated by $S^*_\mathcal{A}$, see \autoref{fig:monotone-assignment}.
We traverse $\mathcal{A}$ in ascending~order, and compare pieces on $\mathcal{B}$ given by the current and an earlier interval on $\mathcal{A}$.
In contrast to previous algorithms~\cite{BuchiNW2025,BuchiBSW2026}, we do a single bottom-up pass on $\mathcal{B}$ instead of repeated top-down~passes.
As it takes $O(N_\mathcal{A})$ time to perform an opposing propagation under $\tnorm[1]$ with this method, lines 8--11 of \autoref{alg:cdtw-apx} take $O(N_{C.\kernedS} + N_{C.\kernedW})$ time in total.
We can sum this bound over all cells.

\begin{lemma}
    \label{thm:monotone-assignment}
    Let $t_0 \in \mathrm{dom}(\mathcal{B})$ for an output border $\mathcal{B}$ with opposing input border $\mathcal{A} := \mathcal{B}.\mathrm{opp}$, and let $\mathcal{A}(s_0)$ be the starting point of a best path from $\mathcal{A}$ to $\mathcal{B}(t_0)$ with smallest $s_0$.
    It is
    \[
        \mathcal{A}.\mathrm{apx}(s_0) + \mathrm{opt}_{\mathcal{A},\mathcal{B}}(s_0,t)
        < \mathcal{A}.\mathrm{apx}(s) + \mathrm{opt}_{\mathcal{A},\mathcal{B}}(s,t)
        \quad \text{whenever } s < s_0 \leq t \text{.}
    \]
    In particular, for all $t \geq s_0$ we have that every parent point $\mathcal{A}(s)$ of $\mathcal{B}(t)$ satisfies $s \geq s_0$.
\end{lemma}

\begin{proof}
    Because $\mathcal{A}(s_0)$ is the starting point of a best path to $\mathcal{B}(t_0)$, \autoref{def:best-path} says $\mathcal{A}.\mathrm{apx}(s_0) + \mathrm{opt}_{\mathcal{A},\mathcal{B}}(s_0,t_0) \leq \mathcal{A}.\mathrm{apx}(s) + \mathrm{opt}_{\mathcal{A},\mathcal{B}}(s,t_0)$ for all $s \leq t_0$.
    If $s < s_0$, this holds with strict inequality since there is no such point $\mathcal{A}(s)$ by choice of $s_0$ as the smallest possible value.
    Using \autoref{thm:happiness-lemma} therefore gives $\mathcal{A}.\mathrm{apx}(s_0) + \rhoinput(s_0) + \rhooutput(t_0) < \mathcal{A}.\mathrm{apx}(s) + \rhoinput(s) + \rhooutput(t_0)$ for all $s < s_0$.

    By subtracting $\rhooutput(t_0)$ from that inequality, adding $\rhooutput(t)$ for a~$t \geq s_0$, and using \autoref{thm:happiness-lemma} once more, we get $\mathcal{A}.\mathrm{apx}(s_0) + \mathrm{opt}_{\mathcal{A},\mathcal{B}}(s_0,t) < \mathcal{A}.\mathrm{apx}(s) + \mathrm{opt}_{\mathcal{A},\mathcal{B}}(s,t)$ for all $s < s_0$ as desired.
    In consequence, there is no parent point $\mathcal{A}(s)$ of $\mathcal{B}(t)$ with $s < s_0$ due to \autoref{thm:apx-factor-opp}.
    \qed
\end{proof}

\begin{proposition}
    \label{thm:propagation-procedure}
    Let $\mathcal{B}$ be an output border with opposing input border $\mathcal{A} := \mathcal{B}.\mathrm{opp}$.
    It is possible to perform $\mathcal{B}.\propagatefromopp$ in $O(N_\mathcal{A})$ time under $\tnorm[1]$.
    \autoref{alg:cdtw-apx} then has a running time in $O(N)$, where $N$ is the total number of quadratic pieces over all cell borders.
\end{proposition}

\begin{proof}
    In the following, we describe a procedure that computes the lower envelope $\mathrm{apx}_{\leq s^*}$ for the subroutine, where $s^* \in \mathrm{dom}(\mathcal{A})$ is the value passed to its call.
    Let $s^*_1, \dotsc, s^*_\lambda$ be the elements of the set $\{ s \in S^*_\mathcal{A} \mid s \leq s^* \}$ in ascending order, where $S^*_\mathcal{A}$ is as in \autoref{thm:parent-candidates}.
    We have~w.l.o.g.\ that the inclusion $s^* \in S^*_\mathcal{A}$ holds, which further implies $s^*_\lambda = s^*$.
    This is since $s^*$ gives a global~and thus also a local minimum:
    Its computation in lines~8--10 of \autoref{alg:cdtw-apx} and the decisions made in lines~4--5 of \autoref{alg:cdtw-subroutines} mean $\mathcal{A}(s^*)$ is a parent point of the north-east cell corner.

    We partition the interval $\{ s \in \mathrm{dom}(\mathcal{A}) \mid s \leq s^*_\lambda \}$, which contains all starting points considered by $\mathrm{apx}_{\leq s^*}$, into subintervals $S_1 := \{ s \in \mathrm{dom}(\mathcal{A}) \mid s \leq s^*_1 \}$ and $S_\kappa := (s^*_{\kappa - 1},s^*_\kappa]$ for $\kappa \in \{2,\dotsc,\lambda\}$.
    Given a subset $T_\kappa \subseteq \mathrm{dom}(\mathcal{B})$ with parent interval $S_\kappa$ for any $\kappa \in \{1,\dotsc,\lambda\}$, \autoref{thm:parent-candidates} implies that the pieces of $\mathrm{apx}_{\leq s^*}$ on $T_\kappa$ are those of the function $\tau_\kappa$ defined by $t \mapsto \mathcal{A}.\mathrm{apx}(t) + \mathrm{opt}_{\mathcal{A},\mathcal{B}}(t,t)$ for $t \in S_\kappa$ and $t \mapsto \mathcal{A}.\mathrm{apx}(s^*_\kappa) + \mathrm{opt}_{\mathcal{A},\mathcal{B}}(s^*_\kappa,t)$ for $t > s^*_\kappa$.
    See \autoref{fig:monotone-assignment} for an example.
    To compute the lower envelope of $\tau_1,\dotsc,\tau_\lambda$, it is not necessary to create all pieces of these functions.
    Instead, we traverse $S_1,\dotsc,S_\lambda$ and create pieces on the fly through maintaining an \emph{active} function.

    Initially, $\tau_1$ is active on all of $S_1$.
    For $\kappa \in \{2,\dotsc,\lambda\}$ we compare the restriction $\tau_\kappa|_{S_\kappa}$ with the yet active function $\tau_\mu$ by checking their overlapping pairs of pieces for intersections as in~\autoref{fig:intersection-of-quadratics}.
    This takes $O(1)$ time per pair.
    If $\tau_\mu$ is nowhere greater than $\tau_\kappa$ on $S_\kappa$, then $\tau_\mu$ remains active~there.
    Otherwise, $\tau_\mu$ only stays active up to the earliest point after which $\tau_\kappa$ is smaller.
    At this point~$\tau_\kappa$ becomes active and stays so on the remainder of $S_\kappa$.
    After $S_\lambda$ has been processed, the traversal is complete and the last active function $\tau_\mu$ stays so on the remaining part of the domain.

    This procedure correctly constructs the lower envelope:
    Let $\mathcal{A}(s_0)$ be the earliest starting point of a best path to an arbitrary $\mathcal{B}(t_0)$, and let $\kappa \in \{1,\dotsc,\lambda\}$ with $s_0 \in S_\kappa$.
    By \autoref{thm:parent-candidates}, all sets of parent candidates are finite.
    Hence, \autoref{thm:monotone-assignment} and continuity of $s \mapsto \mathcal{A}.\mathrm{apx}(s) + \mathrm{opt}_{\mathcal{A},\mathcal{B}}(s,t)$, cf.\ the proof of \autoref{thm:best-paths}, imply that each $t \geq s_0$ has a smallest $s_t \in [s_0,t]$ such that $\mathcal{A}(s_t)$ is a parent candidate for $\mathcal{B}(t)$.
    Together with the definition of $S_\kappa$, we get $\tau_\kappa(t) = \mathcal{A}.\mathrm{apx}(s_t) + \mathrm{opt}_{\mathcal{A},\mathcal{B}}(s_t,t) \leq \mathcal{A}.\mathrm{apx}(s_0) + \mathrm{opt}_{\mathcal{A},\mathcal{B}}(s_0,t) < \min \{ \tau_1(t),\dotsc,\tau_{\kappa-1}(t) \}$ for all $t \geq s_0$, while $\tau_\kappa$ becomes active.

    Furthermore, at most $N_\mathcal{A} + O(|S^*_\mathcal{A}|)$ overlapping pairs of pieces are created under $\tnorm[1]$, which yields a running time in $O(N_\mathcal{A})$ by \autoref{thm:parent-candidates}.
    This is because $\tau_1|_{S_1},\dotsc,\tau_\lambda|_{S_\lambda}$ together have at most $N_\mathcal{A} + |S^*_\mathcal{A}| + O(1)$ pieces, and we create at most $|S^*_\mathcal{A}| + O(1)$ pieces of $\tau_\mu|_{S_\kappa}$ with $\mu < \kappa$ over all active $\tau_\mu$ due to the starting points then being fixed to $s^*_\mu$.
    It remains to update~$\mathcal{B}.\mathrm{apx}$, which already contains $O(1)$ pieces, through taking its lower envelope with the pieces of $\mathrm{apx}_{\leq s^*}$ in $O(N_\mathcal{A})$ time.%
    \footnote{Alternatively to updating $\mathcal{B}.\mathrm{apx}$ at the end, one can instead construct an initial active function $\tau_0$ using its existing pieces, and start comparisons on $S_1$. (Due to the separation explained in the proof of \autoref{thm:inductive-bound-extended}.)}
    Altogether, we perform $\mathcal{B}.\propagatefromopp$ in $O(N_\mathcal{A})$ time, so that \autoref{alg:cdtw-apx} takes $O(N_{C.\kernedS} + N_{C.\kernedW})$ time per cell $C$ by what has been argued above in this section.
    Summing this bound over all cells finally yields a total running time in $O(N)$ under $\tnorm[1]$.
    \qed
\end{proof}

The central challenge is to bound the number $N$ polynomially in the complexities $n,m$ of the input curves.
As indicated by the above proof, \autoref{thm:monotone-assignment} implies $N_\mathcal{B} \leq N_\mathcal{A} + O(|S^*_\mathcal{A}|)$ under~$\tnorm[1]$, but combining that with $|S^*_\mathcal{A}| \in O(N_\mathcal{A})$ from \autoref{thm:parent-candidates} is still too weak.
It cannot rule out an exponential growth over all cells because it allows constants $c > 1$ with $N_\mathcal{B} \geq c \cdot N_\mathcal{A}$.

This necessitates a better bound for $|S^*_\mathcal{A}|$ that does not depend on $N_\mathcal{A}$.
We show that $|S^*_\mathcal{A}|$ grows at worst quadratically over successive propagations between opposing borders.
The resulting inductive bound for $N_\mathcal{B}$ is our main ingredient for establishing that $N$ is indeed polynomial.

\begin{definition}
    \label{def:rank}
    We assign each border $\mathcal{B}$ a \emph{rank}.
    If $\mathcal{B}.\propagatefromopp$ is not called for $\mathcal{B}$, then it has rank $0$.
    Else, it has rank $r := r_0 + 1$, where $r_0 \in \mathbb{N}_0$ is the rank of $\mathcal{B}.\mathrm{opp}$.
\end{definition}

\begin{lemma}
    \label{thm:inductive-bound}
    Let $\mathcal{B}$ be an output border of rank $r > 0$ with opposing input border $\mathcal{A} := \mathcal{B}.\mathrm{opp}$.
    Then the number $N_\mathcal{B}$ is bounded by $N_\mathcal{B} \leq N_\mathcal{A} + O(|S^*_\mathcal{A}|)$ with $|S^*_\mathcal{A}| \in O(r^2)$ under $\tnorm[1]$.
\end{lemma}

\autoref{thm:inductive-bound} follows from the more general \autoref{thm:inductive-bound-extended} that we will prove below in \autoref{subsec:extension}.
In comparison to \cite[Lemma~14]{BuchiNW2025}, the proof of our bounds avoids a distinction between different types of subsegments.
It instead uses the Happiness Lemma~(\autoref{thm:happiness-lemma}) to implicitly handle all of the types that are relevant for our algorithm.
With these bounds in place, it remains to count the number of quadratic pieces on a border as well as the number of borders for each rank.

\begin{proposition}
    \label{thm:rank-bounds}
    In a parameter space with $n \times m$ cells the border ranks are at most $\max\{n,m\}$.
    There are $O(nm)$ borders of rank $0$, each of which has $O(1)$ quadratic pieces under $\tnorm[1]$.
    Moreover, there are $O(nm/r)$ borders of any rank $r > 0$, each of which has $O(r^3)$ pieces under $\tnorm[1]$.
\end{proposition}

\begin{proof}
    The first claim is true since ranks can only accumulate along a row or column of cells by \autoref{def:rank}.
    As there are $O(nm)$ borders in total, this particularly holds for borders of rank $0$, and we have already argued in this section that these borders have $O(1)$ pieces under $\tnorm[1]$.

    Each border $\mathcal{B}$ of rank $r > 0$ has $r$ predecessor borders of ranks $0,\dotsc,r-1$ that are disjoint from all the other predecessors, so the number of rank-$r$ borders is in $O(nm/r)$.
    Finally, we show $N_\mathcal{B} \leq c^2 \cdot r^3$ by induction, where $c \in \mathbb{N}$ is a universal constant for the bounds of \autoref{thm:inductive-bound}.

    In the base case for $r = 1$, this is clear because then only one propagation from a border of rank $0$ with $O(1)$ pieces has occurred.
    For the inductive step we assume $r = r_0 + 1$, where $r_0 > 0$ is the rank of $\mathcal{A} := \mathcal{B}.\mathrm{opp}$.
    By using \autoref{thm:inductive-bound} and the induction hypothesis for $\mathcal{A}$, we get
    \[
        N_\mathcal{B}
        \leq N_\mathcal{A} + c \cdot |S^*_\mathcal{A}|
        \leq c^2 \cdot r_0^3 + c \cdot (c \cdot r^2)
        \leq c^2 \cdot r^2 (r_0 + 1)
        = c^2 \cdot r^3
        \text{.}
        \tag*{\qed}
    \]
\end{proof}

When putting it all together, we assume $n \geq m$ in order to bound the maximum rank by $n$.
This is not a restriction because the input curves $P,Q$ can be swapped:
We have $\mathrm{cdtw}_{\| \cdot \|}(P,Q) = \mathrm{cdtw}_{\| \cdot \|}(Q,P)$ under any norm $\tnorm$, as \autoref{def:cdtw} inherits the norm-induced symmetry.

\begin{theorem}
    \label{thm:running-time-1-norm}
    Given two polygonal curves $P = \langle p_0,\dotsc,p_n \rangle$ and $Q = \langle q_0,\dotsc,q_m \rangle$ in $(\mathbb{R}^2, \tnorm[1])$, where $n \geq m$, our $5$-approximation algorithm for CDTW of $P,Q$ under $\tnorm[1]$ takes $O(n^4 m)$ time.
\end{theorem}

\begin{proof}
    \autoref{thm:propagation-procedure} says that the running time is linear in $N$.
    By \autoref{thm:rank-bounds}, there is a constant $c \in \mathbb{N}$ with $N \leq c \cdot (nm + \sum_{r=1}^{n} nm/r \cdot r^3) = c \cdot nm \cdot (1 + \sum_{r=1}^{n} r^2) \leq 2c \cdot n^4 m$.
    \qed
\end{proof}

\subsection{Extension to Other Norms}
\label{subsec:extension}

We proceed to extend the above implementation and running time analysis of our algorithm from the $1$-norm to general polygonal norms, which we will subsequently use in order to approximate any given norm on $\mathbb{R}^2$ such as the $2$-norm.
For these purposes, we employ the following characterisation:
A function is a norm on $\mathbb{R}^d$ if and only if it is the \emph{gauge} $\mathcal{G}_K \colon \mathbb{R}^d \to \mathbb{R}_{\geq 0}$ of a suitable set~$K \subseteq \mathbb{R}^d$, defined by $\mathcal{G}_K(z) := \inf \{ \lambda \geq 0 \mid z \in \lambda K \}$.
E.g., the $2$-norm is the gauge of the Euclidean unit ball, and the $1$-norm on $\mathbb{R}^2$ is the gauge of the convex polygon that has vertices $(\pm 1,0)^{\mathsf{T}}$ and $(0,\pm 1)^{\mathsf{T}}$.
The precise conditions for $K$ are that it is absorbing, balanced, bounded and convex.
Because the gauge of any such set is equal to the gauges of its closure and interior, one may optionally further require that $K$ is closed or open.
(See \cite[pp.~39--40]{SchaeW1999}.)
We defer formal definitions to later.

\begin{figure}[H]%
    \centering%
    \begin{subfigure}{0.3\linewidth}
        \centering
        \begin{tikzpicture}[x=1.5cm,y=1.5cm]
            \useasboundingbox (-1.5,-1.5) rectangle (1.5,1.5);
            \path[axis] (-1.5,0) -- (1.5,0);
            \path[axis] (0,-1.5) -- (0,1.5);

            \path[tick] (1,4pt) -- (1,-4pt) node[below,inner sep=2.5pt] {$1$};
            \path[tick] (4pt,1) -- (-4pt,1) node[left,inner sep=1.75pt] {$1$};
            \path[tick] (-1,4pt) -- (-1,-4pt);
            \path[tick] (4pt,-1) -- (-4pt,-1);

            \coordinate (v) at (1,0);
            \coordinate (w) at (0,1);
            \coordinate (-v) at ($-1*(v)$);
            \coordinate (-w) at ($-1*(w)$);

            \path[faint fill=UmiBlue] (v) -- (w) -- (-v) -- (-w) -- cycle;
            \path[faint fill=UmiVermillion] (0,0) -- (v) -- (w) -- cycle;
            \path[util grey,rotate around={45:(0.5,0.5)}] ($(0.5,0.5) + 0.125*(1,0)$) |- ($(0.5,0.5) + 0.125*(0,1)$);
            \path[main line,UmiBlue] (v) -- (w) -- (-v) -- node[below left,spacing] {\color{UmiBlue}$K$} (-w) -- cycle;

            \path[vector,UmiVermillion] (0,0) -- node[pos=0.4,below] {\color{UmiVermillion}$v$} (v);
            \path[vector,UmiVermillion] (0,0) -- node[pos=0.4,left,inner sep=2.5pt] {\color{UmiVermillion}$w$} (w);
            \path[vector] (0,0) -- (1,1) node[above,spacing] {$\frac{(w_2 - v_2, v_1 - w_1)^\mathsf{T}}{v_1 w_2 - w_1 v_2}$};
            \node[round point,label={below left:$\mathbf{0}$}] at (0,0) {};
        \end{tikzpicture}
        \subcaption{Tilted square ($1$-norm)}
    \end{subfigure}%
    \hspace*{0.1\linewidth/3}%
    \begin{subfigure}{0.3\linewidth}
        \centering
        \begin{tikzpicture}[x=1.5cm,y=1.5cm]
            \useasboundingbox (-1.5,-1.5) rectangle (1.5,1.5);
            \path[axis] (-1.5,0) -- (1.5,0);
            \path[axis] (0,-1.5) -- (0,1.5);

            \coordinate (v) at (1,-1);
            \coordinate (w) at (1,1);
            \coordinate (-v) at ($-1*(v)$);
            \coordinate (-w) at ($-1*(w)$);

            \path[faint fill=UmiBlue] (v) -- (w) -- (-v) -- (-w) -- cycle;
            \path[faint fill=UmiVermillion] (0,0) -- (v) -- (w) -- cycle;
            \path[main line,UmiBlue] (v) -- (w) -- (-v) -- (-w) -- cycle;

            \path[vector,UmiVermillion] (0,0) -- (v);
            \path[vector,UmiVermillion] (0,0) -- (w);
            \path[vector] (0,0) -- (1,0);
            \node[round point] at (0,0) {};
        \end{tikzpicture}
        \subcaption{Square ($\infty$-norm)}
    \end{subfigure}%
    \hspace*{0.1\linewidth/3}%
    \begin{subfigure}{0.3\linewidth}
        \centering
        \begin{tikzpicture}[x=1.5cm,y=1.5cm]
            \useasboundingbox (-1.5,-1.5) rectangle (1.5,1.5);
            \path[axis] (-1.5,0) -- (1.5,0);
            \path[axis] (0,-1.5) -- (0,1.5);

            \def\vx{1.1}\def\vy{-0.7}
            \def\wx{1.2}\def\wy{0.1}

            \coordinate (v) at (\vx,\vy);
            \coordinate (w) at (\wx,\wy);
            \coordinate (-v) at ($-1*(v)$);
            \coordinate (-w) at ($-1*(w)$);

            \coordinate (v') at (0.9,0.8);
            \coordinate (w') at (-0.2,1.1);
            \coordinate (-v') at ($-1*(v')$);
            \coordinate (-w') at ($-1*(w')$);

            \path[faint fill=UmiBlue] (v) -- (w) -- (v') -- (w') -- (-v) -- (-w) -- (-v') -- (-w') -- cycle;
            \path[faint fill=UmiVermillion] (0,0) -- (v) -- (w) -- cycle;
            \path[main line,UmiBlue] (v) -- (w) -- (v') -- (w') -- (-v) -- (-w) -- (-v') -- (-w') -- cycle;

            \path[vector,UmiVermillion] (0,0) -- (v);
            \path[vector,UmiVermillion] (0,0) -- (w);
            \path[vector] (0,0) -- ($1/(\vx*\wy-\wx*\vy)*(\wy-\vy,\vx-\wx)$);
            \node[round point] at (0,0) {};
        \end{tikzpicture}
        \subcaption{Balanced convex polygon}
    \end{subfigure}%
    \caption{Polygons whose gauges are norms, along with evaluation vectors for specific cones}%
    \label{fig:polygonal-norms}%
\end{figure}
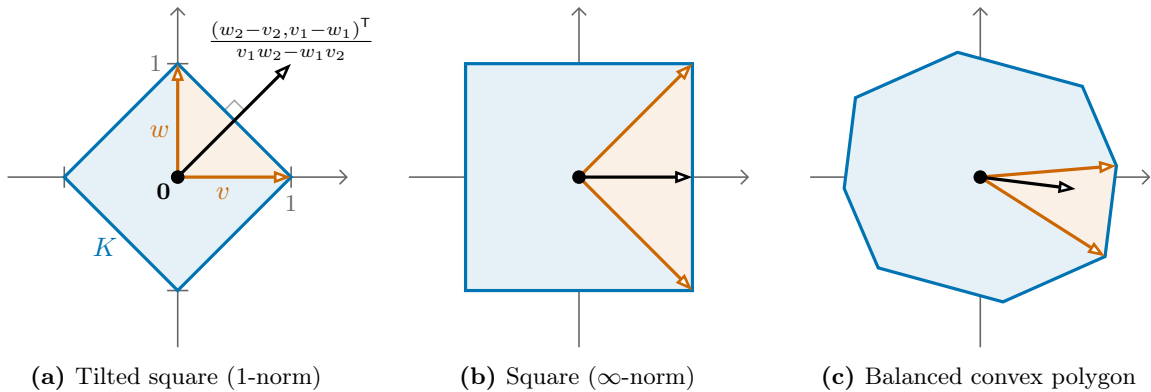

For now, we only need the linearity property of polygonal norms that allows for the extension.
While $\tnorm[1]$ is linear on each quadrant of the plane, any norm $\mathcal{G}_K$ for a suitable polygon $K \subseteq \mathbb{R}^2$ is linear on each cone \smash{$\Delta_{v,w} := \{ \lambda v + \mu w \mid \lambda,\mu \geq 0 \}$}, where $v,w \in \mathbb{R}^2$ are two adjacent vertices of $K$.
In fact, we can evaluate $\mathcal{G}_K$ on the cone $\Delta_{v,w}$ through the scalar product $\mathcal{G}_K(z) = \frac{(w_2 - v_2, v_1 - w_1)}{v_1 w_2 - w_1 v_2} \cdot z$, which uses a vector that is orthogonal to $v - w$.
See \autoref{fig:polygonal-norms} for examples.
This property yields piecewise linear CDTW integrands, and it suffices for proving the generalised inductive bound that we desire.
If $K$ is a polygon of complexity $k \in \mathbb{N}$, i.e.\ it has $k$ vertices and $k$ cones, then~the piecewise quadratic functions that had $O(1)$ pieces under $\tnorm[1]$ have $O(k)$ pieces under $\mathcal{G}_K$.

\begin{lemma}
    \label{thm:inductive-bound-extended}
    Consider any norm $\mathcal{G}_K \colon \mathbb{R}^2 \to \mathbb{R}_{\geq 0}$ such that $K$ is a polygon of complexity $k \in \mathbb{N}$, and let $\mathcal{B}$ be an output border of rank $r > 0$ with opposing input border $\mathcal{A} := \mathcal{B}.\mathrm{opp}$.
    Then the number $N_\mathcal{B}$ is bounded by $N_\mathcal{B} \leq N_\mathcal{A} + O(|S^*_\mathcal{A}| + k)$ with $|S^*_\mathcal{A}| \in O(r^2 \cdot k)$ under $\mathcal{G}_K$.
\end{lemma}

\begin{proof}
    As the considered norm $\mathcal{G}_K$ is linear on $k$ cones partitioning the plane, fixing an argument of $\mathrm{opt}_{\mathcal{A},\mathcal{B}}$ now yields a piecewise quadratic function with $O(k)$ pieces, which is one of $\rhoinput,\rhooutput$ from \autoref{thm:happiness-lemma} plus the fixed offset resulting from the other.
    This is again because the $O(1)$ bending points of the optimal paths from \autoref{thm:optimal-paths} may depend linearly on the non-fixed argument, so that the arguments of $\mathcal{G}_K$ in the integrand for the path costs can switch cones $O(k)$ times.

    For the same reason, the initialisation of $\mathcal{B}.\mathrm{apx}$ in lines~6--7 of \autoref{alg:cdtw-apx} and the subroutine call $\mathcal{B}.\propagatefromadj$ from \autoref{alg:cdtw-subroutines} amount to $O(k)$ quadratic pieces.
    Some of these pieces may remain after the call $\mathcal{B}.\propagatefromopp$, but they cannot result in an increase of complexity:
    Every dominating best path from $\mathcal{A}$ to a point $\mathcal{B}(t)$ intersects all paths from $\mathcal{B}.\mathrm{adj}$ to points $\mathcal{B}(t')$ for $t' > t$.
    By applying \autoref{thm:best-paths} as in the proof of \autoref{thm:apx-factor-opp}, it then follows that there are dominating best paths from $\mathcal{A}$ to all such $\mathcal{B}(t')$.
    Thus, the pieces that the two input borders give are separated, and $\mathcal{B}.\mathrm{adj}$ contributes at most $O(k)$ pieces to $N_\mathcal{B}$.

    To show $N_\mathcal{B} \leq N_\mathcal{A} + O(|S^*_\mathcal{A}| + k)$, we need to bound the number of pieces given by $\mathcal{A} = \mathcal{B}.\mathrm{opp}$.
    For this, we consider the causes that can contribute pieces of $\mathrm{apx}_{\leq s^*}$.
    Recall that we have
    \[
        \mathrm{apx}_{\leq s^*}(t)
        = \min_{s \leq s^* \land s \leq t} \mathcal{A}.\mathrm{apx}(s) + \mathrm{opt}_{\mathcal{A},\mathcal{B}}(s,t)
        = \rhooutput(t) + \min_{s \leq s^* \land s \leq t} \mathcal{A}.\mathrm{apx}(s) + \rhoinput(s)
    \]
    for all $t \in \mathrm{dom}(\mathcal{B})$ due to its definition and \autoref{thm:happiness-lemma}, so one possible cause is a switch to~another piece of~$\rhooutput$.
    Apart from this, \autoref{thm:parent-candidates} implies that any new piece of $\mathrm{apx}_{\leq s^*}$ may only begin when our assignment of parent points needs to switch to another piece of $\mathcal{A}.\mathrm{apx}$ or to one of $\rhoinput$, or when it needs to switch to or from a parent point among the candidates contained in $S^*_\mathcal{A}$.

    These switches never occur in reverse direction since \autoref{thm:monotone-assignment} says that parent points can be assigned monotonically.
    Hence, $\rhoinput$ and $\rhooutput$ contribute $O(k)$ pieces to $N_\mathcal{B}$, while $\mathcal{A}.\mathrm{apx}$ contributes at most $N_\mathcal{A}$ pieces.
    We charge each parent point from $S^*_\mathcal{A}$ twice to cover the case that it causes a new piece both when switching to it and when switching away from it in a monotone assignment.
    This contributes at most $2 \cdot |S^*_\mathcal{A}|$ pieces, so we obtain $N_\mathcal{B} \leq N_\mathcal{A} + O(|S^*_\mathcal{A}| + k)$ as desired.

    We next bound $|S^*_\mathcal{A}|$, which is the number of parent candidates for the north-east cell corner.
    By \autoref{def:parent} and \autoref{thm:happiness-lemma}, these candidates correspond to the semistrict local minima~of~the piecewise quadratic function $s \mapsto \mathcal{A}.\mathrm{apx}(s) + \rhoinput(s)$.
    They can thus only occur at the two domain endpoints, at the breakpoints between pieces, or at points where the derivative is $0$.
    More precisely, a local minimum can occur at some breakpoint only if the left piece's derivative is at most $0$ and the right piece's derivative it at least $0$, see \autoref{fig:local-minima}.
    This is because the function is additionally continuous, cf.\ the proof of \autoref{thm:best-paths}.
    Its left and right derivatives are defined everywhere.

    Let $\partial_{-}$ and $\partial_{+}$ respectively denote the left and right differentiation operator.
    In the following, we split $\mathrm{dom}(\mathcal{A})$ into open intervals such that $\partial_{-} [\mathcal{A}.\mathrm{apx}(s) + \rhoinput(s)] \geq \partial_{+} [\mathcal{A}.\mathrm{apx}(s) + \rhoinput(s)]$ holds on every interval, which implies that local minima attained within them must have derivative~$0$.
    For this, we consider the predecessor borders $\mathcal{A}_0,\dotsc,\mathcal{A}_{r-1}$ of $\mathcal{B}$, where $\mathcal{A}_{r-1} := \mathcal{A}$ has rank $r-1$, and $\mathcal{A}_{i-1} := \mathcal{A}_i.\mathrm{opp}$ has rank $i-1$ for $i \in \{1,\dotsc,r-1\}$.
    They all share the same domain.

    We introduce partition points for our open intervals at the breakpoints of the functions $\mathcal{A}_0.\mathrm{apx}$ and $\rhoinput$ as well as, for all $i \in \{1,\dotsc,r-1\}$, at the breakpoints of $\rhoinput[\mathcal{A}_{i-1}]$ and $\rhooutput[\mathcal{A}_i]$ from \autoref{thm:happiness-lemma}.
    These are $O(r)$ functions and each has $O(k)$ pieces, so their breakpoints split $\mathrm{dom}(\mathcal{A})$ into $O(r \cdot k)$ open intervals.
    In addition to proving the above property, we determine the number of semistrict local minima with derivative $0$ to be at most $r$ on each interval, so we get $|S^*_\mathcal{A}| \in O(r^2 \cdot k)$.

    Let $I \subseteq \mathrm{dom}(\mathcal{A})$ be one of the open intervals.
    As $\mathcal{A}_0.\mathrm{apx}$ has a single piece on $I$ by construction, its left and right derivatives match on all of $I$.
    Given $i \in \{1,\dotsc,r-1\}$, we therefore inductively assume $\partial_{-} \mathcal{A}_{i-1}.\mathrm{apx}(s) \geq \partial_{+} \mathcal{A}_{i-1}.\mathrm{apx}(s)$ for all $s \in I$ and show that $\mathcal{A}_i.\mathrm{apx}$ also has this property.
    For $\mathcal{A}_{r-1} = \mathcal{A}$ it then transfers to $s \mapsto \mathcal{A}.\mathrm{apx}(s) + \rhoinput(s)$ since $\rhoinput$ has a single piece on $I$.

    As for the inductive step:
    \autoref{thm:parent-candidates} implies that $\mathcal{A}_i.\mathrm{apx}|_I$ is the lower envelope of
    \begin{gather*}
        \alpha_i \colon I \to \mathbb{R}_{\geq 0}
        \text{,} \quad \text{defined by} \quad
        \alpha_i(t) := \mathcal{A}_{i-1}.\mathrm{apx}(t) + \rhoinput[\mathcal{A}_{i-1}](t) + \rhooutput[\mathcal{A}_{i}](t)
        \text{,} \\[3pt]
        \text{as well as} \quad
        \alpha_{i,s_0} \colon \{ t \in I \mid t \geq s_0 \} \to \mathbb{R}_{\geq 0}
        \quad \text{ for all } s_0 \in S^*_{\mathcal{A}_{i-1}}
        \text{,} \\[3pt]
        \text{defined by} \quad
        \alpha_{i,s_0}(t) := \mathcal{A}_{i-1}.\mathrm{apx}(s_0) + \rhoinput[\mathcal{A}_{i-1}](s_0) + \rhooutput[\mathcal{A}_{i}](t)
        \text{.}
    \end{gather*}
    Due to $\rhoinput[\mathcal{A}_{i-1}]$ and $\rhooutput[\mathcal{A}_{i}]$ each having a single piece on $I$ by construction, it follows $\partial_- \alpha_i(t) \geq \partial_+ \alpha_i(t)$ for all $t \in I$ via the induction hypothesis.
    Meanwhile, for all $s_0 \in S^*_{\mathcal{A}_{i-1}}$ and all $t \in I$ with $t > s_0$ we even have $\partial_- \alpha_{i,s_0}(t) = \partial_+ \alpha_{i,s_0}(t)$ since the argument of $\mathcal{A}_{i-1}.\mathrm{apx}$ in $\alpha_{i,s_0}$ is fixed to $s_0$.

    The lower envelope retains this property:
    It only switches from one of the functions to another if the next function's right derivative is less than or equal to the previous one's.
    Else, the previous function would still yield smaller values beyond the breakpoint, see \autoref{fig:lower-envelope-and-derivatives}.
    Thus, $\partial_{-} \mathcal{A}_{i}.\mathrm{apx}(t) \geq \allowbreak \partial_{+} \mathcal{A}_{i}.\mathrm{apx}(t)$ holds for all $t \in I$, including the breakpoints introduced by the lower envelope.

    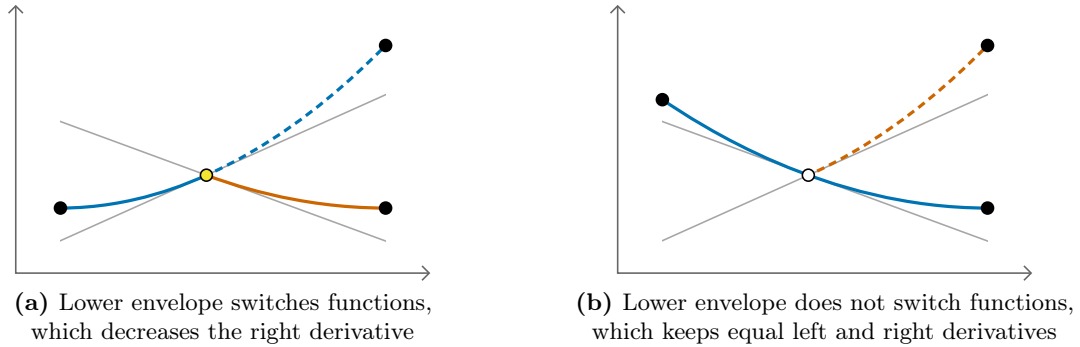
\begin{figure}[H]%
        \centering%
        \begin{subfigure}{0.45\linewidth}
            \centering
            \begin{tikzpicture}[x=2.15cm,y=2.15cm]
                \useasboundingbox (-0.275,1.25) rectangle (2.275,-0.4);
                \path[axes] (-0.275,1.25) |- (2.275,-0.4);

                \path[name path=a] (0,0) parabola[bend at start] (2,1);
                \path[name path=b] (0,2/3) parabola[bend at end] (2,0);
                \path[name intersections={of=a and b,name=ab}];

                \path[main line,UmiBlue] (0,0) parabola[bend at start] (ab-1);
                \path[main line,UmiVermillion] (ab-1) parabola[bend at end] (2,0);

                \begin{scope}
                    \clip (ab-1) rectangle (2,1);
                    \path[main line,UmiBlue,util dash] (ab-1) parabola bend (0,0) (2,1);
                \end{scope}

                \node[round point] at (0,0) {};
                \node[round point,fill=UmiYellow] at (ab-1) {};
                \node[round point] at (2,0) {};
                \node[round point] at (2,1) {};

                \begin{scope}[on background layer]
                    \clip (0,1.125) rectangle (2,-0.35);
                    \path[util grey]
                        let \p1=(ab-1),\n1={1/2*\x1}
                        in ($(ab-1) - 2*(1,\n1)$) -- ($(ab-1) + 2*(1,\n1)$);
                    \path[util grey]
                        let \p1=(ab-1),\p2=(1,0),\n1={1/3*(\x1-2*\x2)}
                        in ($(ab-1) - 2*(1,\n1)$) -- ($(ab-1) + 2*(1,\n1)$);
                \end{scope}
            \end{tikzpicture}
            \subcaption{Lower envelope switches functions, \\ which decreases the right derivative}
        \end{subfigure}%
        \hspace*{0.05\linewidth}%
        \begin{subfigure}{0.45\linewidth}
            \centering
            \begin{tikzpicture}[x=2.15cm,y=2.15cm]
                \useasboundingbox (-0.275,1.25) rectangle (2.275,-0.4);
                \path[axes] (-0.275,1.25) |- (2.275,-0.4);

                \path[name path=a] (0,0) parabola[bend at start] (2,1);
                \path[main line,UmiBlue,name path=b] (0,2/3) parabola[bend at end] (2,0);
                \path[name intersections={of=a and b,name=ab}];

                \begin{scope}
                    \clip (ab-1) rectangle (2,1);
                    \path[main line,UmiVermillion,util dash] (ab-1) parabola bend (0,0) (2,1);
                \end{scope}

                \node[round point] at (0,2/3) {};
                \node[round point,fill=white] at (ab-1) {};
                \node[round point] at (2,0) {};
                \node[round point] at (2,1) {};

                \begin{scope}[on background layer]
                    \clip (0,1.125) rectangle (2,-0.35);
                    \path[util grey]
                        let \p1=(ab-1),\n1={1/2*\x1}
                        in ($(ab-1) - 2*(1,\n1)$) -- ($(ab-1) + 2*(1,\n1)$);
                    \path[util grey]
                        let \p1=(ab-1),\p2=(1,0),\n1={1/3*(\x1-2*\x2)}
                        in ($(ab-1) - 2*(1,\n1)$) -- ($(ab-1) + 2*(1,\n1)$);
                \end{scope}
            \end{tikzpicture}
            \subcaption{Lower envelope does not switch functions, \\ which keeps equal left and right derivatives}
        \end{subfigure}%
        \caption{Behaviour of a continuous lower envelope depending on the functions' derivatives}%
        \label{fig:lower-envelope-and-derivatives}%
    \end{figure}

    It remains to bound the number of semistrict local minima with derivative $0$ occurring on $I$.
    For this, we count the distinct $(a,b)$-coefficient pairs over all the quadratic pieces $s \mapsto as^2 + bs + c$ of the function $\mathcal{A}.\mathrm{apx}$ on $I$.
    Such a pair determines the derivative of the pieces to be $s \mapsto 2as + b$, and thereby fixes the the sole point with derivative $0$ to be $-b/(2a)$ in case of $a \neq 0$.
    That means the number of distinct $(a,b)$-pairs gives an upper bound for the number of parabola vertices.

    Note that the case of $a = 0$ does not contribute additional semistrict local minima, as it~yields a derivative of constant value $b$.
    Local minima on $I$ have derivative $0$ as argued before, requiring also $b = 0$ and thus constant-valued pieces.
    Any semistrict minimum on those pieces can only be attained at a piece breakpoint.
    Such a point then has to coincide either with an endpoint of the open interval $I$, or with a point of derivative $0$ that is fixed by a neighbouring piece's pair.

    We use induction again to prove that $\mathcal{A}_i.\mathrm{apx}$ has at most $i+1$ distinct $(a,b)$-coefficient pairs on $I$ for all $i \in \{0,\dotsc,r-1\}$.
    The base case for $i=0$ is true since $\mathcal{A}_0.\mathrm{apx}$ has a single piece on~$I$.
    For $i > 0$ we inductively assume that $\mathcal{A}_{i-1}.\mathrm{apx}$ has at most $i$ such pairs, and we consider $\mathcal{A}_i.\mathrm{apx}$ in form of the lower envelope from above.
    Adding $\rhoinput[\mathcal{A}_{i-1}]$ and $\rhooutput[\mathcal{A}_{i}]$, which both have a single piece on $I$, to $\mathcal{A}_{i-1}.\mathrm{apx}$ preserves the number of $(a,b)$-pairs.
    Therefore, $\alpha_i$ has at most $i$ pairs on $I$.

    Moreover, all functions $\alpha_{i,s_0}$ for $s_0 \in S^*_{\mathcal{A}_{i-1}}$ share the single $(a,b)$-pair of $\rhooutput[\mathcal{A}_{i}]$ on $I$ because the terms with fixed arguments contribute to $c$ instead of $a$ and $b$.
    Together, we have that $\mathcal{A}_i.\mathrm{apx}$ has at most $i+1$ distinct $(a,b)$-coefficient pairs on $I$.
    For $\mathcal{A}_{r-1} = \mathcal{A}$ the number $r$ is again preserved by adding the single piece of $\rhoinput$ on $I$, so that we overall have $|S^*_\mathcal{A} \cap I| \leq r$ as claimed above.
    \qed
\end{proof}

Note that the ideas of counting $(a,b)$-coefficient pairs \cite[Lemma~14]{BuchiNW2025} as well as utilising left and right derivatives \cite[Lemma~17]{BuchiNW2025} originate from the analysis of the exact 1D algorithm.
There, the desired properties and bounds were shown separately on subsegments, which were associated with different types of optimal paths as in \autoref{fig:optimal-path-types}.
This required a distinction into a lot of cases, and it caused the analysis to become large and repetitive.
One difficulty was the lack of tools for taking a unified approach and exploiting common principles behind function propagations.

\autoref{thm:happiness-lemma} now provides such a tool.
Although the subsegments are still reflected by the above proof splitting the domain into intervals, our analysis is much shorter and more straightforward than that from \cite{BuchiNW2025}.
We believe that this constitutes an insightful contribution to the understanding of propagation complexity.
A general analysis of unhappy paths in 2D remains open for future~work.
The costs of these paths cannot be split into two univariate functions (cf.~\cite[Section~4.5]{BuchiNW2025}).

We next use \autoref{thm:inductive-bound-extended} to establish that under a polygonal norm $\mathcal{G}_K$ of complexity $k \in \mathbb{N}$~our algorithm can be implemented with running time $O(n^4 m \cdot k)$.
This extension is not difficult, yet there are two important aspects for achieving a linear factor in $k$ instead of a worse factor:

First, the fact that the procedure for opposing propagation from \autoref{thm:propagation-procedure} uses a single pass on the output border $\mathcal{B}$, which is an algorithmic contribution of ours.
Doing multiple passes as in \cite{BuchiNW2025,BuchiBSW2026} would lead to a higher running time (cf.~\cite[Proposition~18]{BuchiBSW2026}).
Second, evaluating $\mathcal{G}_K$ at a point $z \in \mathbb{R}^2$ requires finding the cone of $K$ that contains $z$.
This is possible in $O(\log(k))$ time via binary search.
A naive implementation would do this $O(k)$ times per cell to compute costs of optimal paths, which we avoid by keeping track of how the arguments of $\mathcal{G}_K$ switch cones.

\begin{theorem}
    \label{thm:running-time-polygonal-norm}
    Consider any norm $\mathcal{G}_K \colon \mathbb{R}^2 \to \mathbb{R}_{\geq 0}$ such that $K$ is a polygon of complexity $k \in \mathbb{N}$.
    Given two polygonal curves $P = \langle p_0,\dotsc,p_n \rangle$ and $Q = \langle q_0,\dotsc,q_m \rangle$ in $(\mathbb{R}^2, \mathcal{G}_K)$, where $n \geq m$, our $5$-approximation algorithm for CDTW of~$P,Q$ under $\mathcal{G}_K$ takes $O(n^4m \cdot k)$ time.
\end{theorem}

\begin{proof}
    As explained in the proof of \autoref{thm:inductive-bound-extended}, under $\mathcal{G}_K$ we have that lines~1--2 and lines~6--7 of \autoref{alg:cdtw-apx} yield $O(k)$ quadratic pieces on every initialised border, so $O(nm \cdot k)$ pieces overall.
    These can be computed within the same time bound:
    By \autoref{def:path-and-cost}, the cost of a straight path is a definite integral of some piecewise linear function $\mathcal{G}_K \circ \psi$, where $\psi \colon \mathbb{R} \to \mathbb{R}^2$ is an affine map.
    To~evaluate this integrand at any initial $t \in \mathbb{R}$, we can find two adjacent vertices $v,w \in \mathbb{R}^2$ of $K$ via binary search such that $\psi(t)$ is contained in the cone $\smash{\Delta_{v,w}} := \{ \lambda v + \mu w \mid \lambda,\mu \geq 0 \}$.

    Through integration, the linear piece provided by the scalar product $t \mapsto \frac{(w_2 - v_2, v_1 - w_1)}{v_1 w_2 - w_1 v_2} \cdot \psi(t)$ gives an initial quadratic piece.
    A breakpoint of $\mathcal{G}_K \circ \psi$ occurs when $\psi$ switches from $\Delta_{v,w}$ to a different cone of $K$.
    We can determine that cone in $O(1)$ time by intersecting the line traced by the affine map $\psi$ with the lines that bound $\Delta_{v,w}$.
    Thus, we perform binary search once for the initial piece in $O(\log(k))$ time, and compute the other $O(k)$ pieces in $O(1)$ time each.

    For every cell $C$ we further compute an $\ell$ as in \autoref{thm:optimal-paths}, which gives optimal paths within~$C$.
    This can also be done via binary search on the vertices of $K$ (see \cite[Section~2]{BuchiBSW2026}), taking $O(\log(k))$ time per cell $C$.
    The next step is computing the parent candidates for the north-east corner of~$C$.
    There exist $|S^*_\mathcal{A}| \in O(N_\mathcal{A} + k)$ candidates on each input border $\mathcal{A}$ because the respective function from lines~8--9 of \autoref{alg:cdtw-apx} has $N_\mathcal{A} + O(k)$ quadratic pieces under $\mathcal{G}_K$.
    To construct the pieces, we again use the above technique for determining cones of $K$.
    Here, we apply it to the non-fixed starting points of the optimal paths from \autoref{thm:optimal-paths}, and to their $O(1)$ bending points.

    It follows that computing $S^*_\mathcal{A}$ takes $O(N_\mathcal{A} + k)$ time.
    By checking every parent candidate, we thus obtain the required values from lines~10--11 in running time $O(N_{C.\kernedS} + N_{C.\kernedW} + k)$ per cell $C$.
    We claim that this time bound also applies to the subroutines when using the above technique for determining cones of $K$.
    In $\mathcal{B}.\propagatefromadj$ the function $\mathrm{apx}_{s^*}$ has $O(k)$ pieces under $\mathcal{G}_K$ due to the fixed starting point.
    Updating $\mathcal{B}.\mathrm{apx}$ through taking the lower envelope of its already existing $O(k)$ pieces with those of $\mathrm{apx}_{s^*}$ takes $O(k)$ time.
    Note that the two types of pieces are separated in the lower envelope, cf.\ the second property from \autoref{thm:apx-function-properties}.

    Similarly, we also have such a separation for opposing propagation, as explained in the proof of \autoref{thm:inductive-bound-extended}.
    We can therefore easily update $\mathcal{B}.\mathrm{apx}$ during $\mathcal{B}.\propagatefromopp$, after we have computed $\mathrm{apx}_{\leq s^*}$.
    For the latter we use the procedure from the proof of \autoref{thm:propagation-procedure}.
    Under $\mathcal{G}_K$ the restrictions $\tau_1|_{S_1},\dotsc,\tau_\lambda|_{S_\lambda}$ together have at most $N_\mathcal{A} + |S^*_\mathcal{A}| + O(k)$ pieces, where $\mathcal{A} := \mathcal{B}.\mathrm{opp}$, and we create at most $|S^*_\mathcal{A}| + O(k)$ pieces of $\tau_\mu|_{S_\kappa}$ with $\mu < \kappa$ over all active $\tau_\mu$, so the number of overlapping pairs of pieces is at most $N_\mathcal{A} + 2 \cdot |S^*_\mathcal{A}| + c \cdot k \in O(N_\mathcal{A} + k)$.

    Hence, we indeed spend $O(N_{C.\kernedS} + N_{C.\kernedW} + k)$ time per cell $C$, and summing this over all cells yields $O(N + nm \cdot k)$.
    We use \autoref{thm:inductive-bound-extended} to bound $N$.
    Proceeding analogously to \autoref{thm:rank-bounds} gives $N_\mathcal{B} \leq c^2 \cdot r^3 \cdot k + c \cdot r \cdot k \in O(r^3 \cdot k)$ for each border $\mathcal{B}$ of any rank $r > 0$, while borders of rank~$0$ have $O(k)$ pieces under $\mathcal{G}_K$.
    We then obtain $N \in O(n^4m \cdot k)$ analogously to \autoref{thm:running-time-1-norm}, which dominates the term $O(nm \cdot k)$.
    Overall, the running time is thus in $O(n^4m \cdot k)$.
    \qed
\end{proof}

To approximate CDTW in 2D under the $2$-norm, we may use a regular polygon inscribed in the Euclidean unit disk.
This was done in \cite[Corollary~17]{BuchiBSW2026}, and a factor of $1+\varepsilon$ requires polygon complexity $k \in O(\varepsilon^{-1/2})$.
In fact, results from convex geometry imply that we can approximate any fixed norm on $\mathbb{R}^2$ by some polygonal norm of the same complexity.
More generally, any fixed norm on $\mathbb{R}^d$ is $(1+\varepsilon)$-approximated by some polyhedral norm of complexity $O(\varepsilon^{(1-d)/2})$.

In the following, we outline this polyhedral approximation and its application to our results.
Its first ingredient is the geometric characterisation of norms, which we have introduced before:
A function is a norm on $\mathbb{R}^d$ if and only if it is the gauge $\mathcal{G}_K$ of an absorbing, balanced, bounded, convex and closed set $K$.
The first two properties are defined as follows (see \cite[p.~11]{SchaeW1999}).
\begin{enumerate}
    \item
    A set $K \subseteq \mathbb{R}^d$ is called \emph{absorbing} (or \emph{radial}) if for every point $z \in \mathbb{R}^d$ there exists some $\lambda_0 \geq 0$ such that $z \in \lambda K$ holds whenever $|\lambda| \geq \lambda_0$.
    This ensures that $\mathcal{G}_K$ indeed assigns every $z \in \mathbb{R}^d$ a real number $\mathcal{G}_K(z) = \inf \{ \lambda \geq 0 \mid z \in \lambda K \}$.
    It further implies the inclusion $\mathbf{0} \in K$.

    \item
    A set $K \subseteq \mathbb{R}^d$ is called \emph{balanced} (or \emph{circled}) in case of $\lambda K \subseteq K$ for all $\lambda \in [-1,1]$.
    This yields absolute homogeneity of $\mathcal{G}_K$, i.e.\ that $\mathcal{G}_K(\lambda z) = |\lambda| \mathcal{G}_K(z)$ holds for all $z \in \mathbb{R}^d$ and all $\lambda \in \mathbb{R}$.
    If $K$ is a polyhedron, it further implies that $K$ must have an even number of vertices.
\end{enumerate}

The other three properties are defined as usual, and we do not require those definitions~here.
Every absorbing, convex and closed set in $\mathbb{R}^d$ contains the origin $\mathbf{0}$ in its interior.
When a~bounded, convex and closed set in $\mathbb{R}^d$ has a non-empty interior, then it is also called a \emph{convex body}.
Moreover, gauges allow us to describe any norm $\tnorm$ on $\mathbb{R}^d$ through its unit sublevel set:
We have that $\tnorm$ is the gauge $\mathcal{G}_M$ of the balanced convex body $M := \{ z \in \mathbb{R}^d \mid \| z \| \leq 1 \}$.
(See \cite[p.~40]{SchaeW1999}.)

The second ingredient is the approximation of convex bodies by (convex) polyhedra, for~which there is a vast amount of literature, see \cite{Brons2008} for a survey.
Note that a $d$-dimensional polyhedron $K$ generally has different numbers of vertices and facets when $d > 2$, where the facets correspond to the cones on which the gauge $\mathcal{G}_K$ is linear.
The specific result that we need holds independent of whether one desires a bound on the number of vertices or facets, though this may lead to~distinct polyhedra.
Because we do not consider the equipped norm to be an input for CDTW, we suppress all constants in the size and dimension of the given convex body by treating it as fixed.

\begin{lemma}[{\cite[Section~4.4]{Brons2008}}]
    \label{thm:convex-apx}
    Fix any balanced convex body $M \subseteq \mathbb{R}^d$.
    For every $\varepsilon > 0$ there is a convex polyhedron $K \subseteq \mathbb{R}^d$ with $O(\varepsilon^{(1-d)/2})$ facets (or vertices) such that $M \subseteq K \subseteq (1+\varepsilon) M$.
\end{lemma}

For e.g.\ the sublevel sets of $p$-norms, which are hyperellipsoids if $1 < p < \infty$, it is not difficult to determine these approximating polyhedra.
Constructions for the general case may involve very large constants, see \cite[Section~8]{Brons2008} for algorithmic considerations.
A simple approach, albeit with a higher worst-case asymptotic complexity, is to choose roughly $\varepsilon^{-d} \log(\varepsilon^{-d})$ points uniformly at random from the boundary of $(1+\varepsilon) M$.
Then their convex hull is a $(1+\varepsilon)$-approximation of $M$ with high probability \cite{Naszo2019}.
Also, note that the following proposition might be applicable to other results that work for polyhedral norms, such as the subquadratic DTW algorithm from \cite{GoldS2018}.

\begin{proposition}
    \label{thm:norm-apx}
    Fix any norm $\tnorm$ on $\mathbb{R}^d$.
    For every $\varepsilon > 0$ there is another norm $\mathcal{G}_K$ on $\mathbb{R}^d$, where $K \subseteq \mathbb{R}^d$ is a balanced convex polyhedron with $O(\varepsilon^{(1-d)/2})$ facets (or vertices), such that
    \[
        \| z \| \leq \mathcal{G}_K(z) \leq (1 + \varepsilon) \| z \|
        \quad \text{holds for all } z \in \mathbb{R}^d \text{.}
    \]
\end{proposition}

\begin{proof}
    Since $\tnorm$ is a norm, the set $M := \{ z \in \mathbb{R}^d \mid \| z \| \leq 1 \}$ is a balanced convex body in $\mathbb{R}^d$.
    We use \autoref{thm:convex-apx} to get a convex polyhedron $K_0$ with a bounded number of facets (or vertices) and $M \subseteq K_0 \subseteq (1+\varepsilon) M$.
    Let $K$ be the \emph{disked core (or hull)} of $(1+\varepsilon)^{-1}K_0$, which is its largest balanced and convex subset (or smallest such superset).
    It follows $(1+\varepsilon)^{-1} M \subseteq K \subseteq M$,~so
    \[
        \| z \|
        = \mathcal{G}_{M}(z)
        \leq \mathcal{G}_{K}(z)
        \leq \mathcal{G}_{(1+\varepsilon)^{-1} M}(z)
        = \mathcal{G}_M((1 + \varepsilon)z)
        = (1 + \varepsilon) \mathcal{G}_M(z)
        = (1 + \varepsilon) \| z \|
    \]
    holds for all $z \in \mathbb{R}^d$ by definition of $M$.
    The gauge $\mathcal{G}_K$ is a norm, as $K$ is a balanced convex~body.
    Finally, the disked core (or hull) $K$ can have at most twice as many facets (or vertices) as $K_0$.
    \qed
\end{proof}

\begin{theorem}
    \label{thm:running-time-any-norm}
    Fix any norm $\tnorm$ on $\mathbb{R}^2$, and let $\varepsilon > 0$ be arbitrary.
    Given two polygonal curves $P = \langle p_0,\dotsc,p_n \rangle$ and $Q = \langle q_0,\dotsc,q_m \rangle$ in $(\mathbb{R}^2,\tnorm)$, where $n \geq m$, we can compute a factor-$(5+\varepsilon)$ approximation for CDTW of $P,Q$ under $\tnorm$ in running time $O(n^4m / \varepsilon^{1/2})$.
\end{theorem}

\begin{proof}
    We use \autoref{thm:norm-apx} for $\varepsilon_0 := \min \{ \frac{\varepsilon}{15}, 1 \}$ to get a norm $\mathcal{G}_K$.
    By \autoref{def:cdtw}, it is
    \[
        \mathrm{cdtw}_{\| \cdot \|}(P,Q)
        \leq \mathrm{cdtw}_{\mathcal{G}_{K}}(P,Q)
        \leq (1 + \varepsilon_0)^2 \cdot \mathrm{cdtw}_{\| \cdot \|}(P,Q)
        \leq (1 + \tfrac{\varepsilon}{5}) \cdot \mathrm{cdtw}_{\| \cdot \|}(P,Q)
        \text{.}
    \]
    Applying \autoref{thm:running-time-polygonal-norm} to $\mathcal{G}_K$ therefore returns a $(5+\varepsilon)$-approximation in $O(n^4 m / \varepsilon^{1/2})$ time.
    \qed
\end{proof}

\section{Conclusion}

We established the first constant-factor approximation algorithm for CDTW in 2D that has a polynomial running time.
It relies on a new integration-based building block and uses a tailored propagation scheme.
Another technical contribution of ours is a unified approach to the running time analysis, thanks to novel insights into core principles behind propagation complexity.

It is still open whether the additional propagation patterns occurring in exact 2D algorithms, which we circumvented, have polynomial complexity.
We also do not know if the approximation factor of $5$ is tight, see \autoref{rem:tightness} and \autoref{subsec:apx-factor}.
To evaluate efficiency and solution quality in practice, an experimental comparison of an exact algorithm and ours would be useful.

Our algorithm works for all polygonal norms.
We demonstrated that this allows to approximate CDTW in 2D under any desired norm, making our results applicable in a very broad 2D setting.
Recently, there has been an increased interest in solving not only geometric but also combinatorial optimisation problems under more general norm objective functions, see for example \cite{ChakrS2019,ChenLRZ2025}.

One advantage of measuring the similarity of curves under a norm of choice is the possibility to assign the curve dimensions, which may carry different types of data, suitable weightings~\cite{GutscS2022}.
Thus, higher dimensions are a natural direction for future work on CDTW.
Known fundamentals do not suffice there since \autoref{thm:optimal-paths} does not hold in 3D and beyond, see \autoappref{app:3d-example}.

    \bibliography{literature.bib}

\begin{thebibliography}{10}

\bibitem{BrakaPSW2005}
Sotiris Brakatsoulas, Dieter Pfoser, Randall Salas, and Carola Wenk.
\newblock On map-matching vehicle tracking data.
\newblock In {\em Proceedings of the 31st International Conference on Very
  Large Data Bases}, pages 853--864. VLDB Endowment, 2005.
\newblock URL: \url{https://dl.acm.org/doi/10.5555/1083592.1083691}.

\bibitem{Brank2022}
Milutin Brankovic.
\newblock {\em Graphs and Trajectories in Practical Geometric Problems}.
\newblock PhD thesis, University of Sydney, 2022.

\bibitem{BrankBKNPW2020}
Milutin Brankovic, Kevin Buchin, Koen Klaren, André Nusser, Aleksandr Popov,
  and Sampson Wong.
\newblock $(k,l)$-medians clustering of trajectories using {Continuous Dynamic
  Time Warping}.
\newblock In {\em Proceedings of the 28th International Conference on Advances
  in Geographic Information Systems}, pages 99--110. ACM, 2020.
\newblock \href {https://doi.org/10.1145/3397536.3422245}
  {\path{doi:10.1145/3397536.3422245}}.

\bibitem{Brons2008}
Efim~M. Bronstein.
\newblock Approximation of convex sets by polytopes.
\newblock {\em Journal of Mathematical Sciences}, 153(6):727--762, 2008.
\newblock \href {https://doi.org/10.1007/S10958-008-9144-X}
  {\path{doi:10.1007/S10958-008-9144-X}}.

\bibitem{BuchiBSW2026}
Kevin Buchin, Maike Buchin, Jan~Erik Swiadek, and Sampson Wong.
\newblock Fundamentals of computing {Continuous Dynamic Time Warping} in {2D}
  under different norms.
\newblock In {\em 20th International Conference and Workshops on Algorithms and
  Computation}, pages 467--482. Springer, 2026.
\newblock \href {https://arxiv.org/abs/2511.20420} {\path{arXiv:2511.20420}},
  \href {https://doi.org/10.1007/978-981-95-7127-7_31}
  {\path{doi:10.1007/978-981-95-7127-7_31}}.

\bibitem{BuchiBW2009}
Kevin Buchin, Maike Buchin, and Yusu Wang.
\newblock Exact algorithms for partial curve matching via the {Fréchet}
  distance.
\newblock In {\em Proceedings of the Twentieth Annual ACM-SIAM Symposium on
  Discrete Algorithms}, pages 645--654. SIAM, 2009.
\newblock \href {https://doi.org/10.1137/1.9781611973068.71}
  {\path{doi:10.1137/1.9781611973068.71}}.

\bibitem{BuchiNW2025}
Kevin Buchin, André Nusser, and Sampson Wong.
\newblock Computing {Continuous Dynamic Time Warping} of time series in
  polynomial time.
\newblock {\em Journal of Computational Geometry}, 16(1):765--799, 2025.
\newblock \href {https://doi.org/10.20382/JOCG.V16I1A21}
  {\path{doi:10.20382/JOCG.V16I1A21}}.

\bibitem{Buchi2007}
Maike Buchin.
\newblock {\em On the Computability of the {Fréchet} Distance between
  Triangulated Surfaces}.
\newblock PhD thesis, Freie Universität Berlin, 2007.
\newblock URL: \url{https://refubium.fu-berlin.de/handle/fub188/1909}.

\bibitem{ChakrS2019}
Deeparnab Chakrabarty and Chaitanya Swamy.
\newblock Approximation algorithms for minimum norm and ordered optimization
  problems.
\newblock In {\em Proceedings of the 51st Annual ACM SIGACT Symposium on Theory
  of Computing}, pages 126--137. ACM, 2019.
\newblock \href {https://doi.org/10.1145/3313276.3316322}
  {\path{doi:10.1145/3313276.3316322}}.

\bibitem{ChenLRZ2025}
Kuowen Chen, Jian Li, Yuval Rabani, and Yiran Zhang.
\newblock New results on a general class of minimum norm optimization problems.
\newblock In {\em 52nd International Colloquium on Automata, Languages, and
  Programming}, pages 50:1--50:20. Schloss Dagstuhl -- Leibniz-Zentrum für
  Informatik, 2025.
\newblock \href {https://doi.org/10.4230/LIPICS.ICALP.2025.50}
  {\path{doi:10.4230/LIPICS.ICALP.2025.50}}.

\bibitem{EfratFV2007}
Alon Efrat, Quanfu Fan, and Suresh Venkatasubramanian.
\newblock Curve matching, time warping, and light fields: New algorithms for
  computing similarity between curves.
\newblock {\em Journal of Mathematical Imaging and Vision}, 27(3):203--216,
  2007.
\newblock \href {https://doi.org/10.1007/S10851-006-0647-0}
  {\path{doi:10.1007/S10851-006-0647-0}}.

\bibitem{GoldS2018}
Omer Gold and Micha Sharir.
\newblock {Dynamic Time Warping} and {Geometric Edit Distance}: Breaking the
  quadratic barrier.
\newblock {\em ACM Transactions on Algorithms}, 14(4):50:1--50:17, 2018.
\newblock \href {https://doi.org/10.1145/3230734} {\path{doi:10.1145/3230734}}.

\bibitem{GutscS2022}
Theodor Gutschlag and Sabine Storandt.
\newblock On the generalized {Fréchet} distance and its applications.
\newblock In {\em Proceedings of the 30th International Conference on Advances
  in Geographic Information Systems}, pages 35:1--35:10. ACM, 2022.
\newblock \href {https://doi.org/10.1145/3557915.3560970}
  {\path{doi:10.1145/3557915.3560970}}.

\bibitem{Har-PRR2025}
Sariel Har-Peled, Benjamin Raichel, and Eliot~W. Robson.
\newblock The {Fréchet} distance unleashed: Approximating a dog with a frog.
\newblock In {\em 41st International Symposium on Computational Geometry},
  pages 54:1--54:13. Schloss Dagstuhl -- Leibniz-Zentrum für Informatik, 2025.
\newblock \href {https://arxiv.org/abs/2407.03101} {\path{arXiv:2407.03101}},
  \href {https://doi.org/10.4230/LIPICS.SOCG.2025.54}
  {\path{doi:10.4230/LIPICS.SOCG.2025.54}}.

\bibitem{Klare2020}
Koen Klaren.
\newblock {Continuous Dynamic Time Warping} for clustering curves.
\newblock Master's thesis, Eindhoven University of Technology, 2020.
\newblock URL:
  \url{https://research.tue.nl/en/studentTheses/continuous-dynamic-time-warping-for-clustering-curves}.

\bibitem{MahesSS2018}
Anil Maheshwari, Jörg-Rüdiger Sack, and Christian Scheffer.
\newblock Approximating the integral {Fréchet} distance.
\newblock {\em Computational Geometry}, 70--71:13--30, 2018.
\newblock \href {https://doi.org/10.1016/J.COMGEO.2018.01.001}
  {\path{doi:10.1016/J.COMGEO.2018.01.001}}.

\bibitem{MunicP1999}
Mario~E. Munich and Pietro Perona.
\newblock {Continuous Dynamic Time Warping} for translation-invariant curve
  alignment with applications to signature verification.
\newblock In {\em Proceedings of the Seventh IEEE International Conference on
  Computer Vision}, pages 108--115. IEEE, 1999.
\newblock \href {https://doi.org/10.1109/ICCV.1999.791205}
  {\path{doi:10.1109/ICCV.1999.791205}}.

\bibitem{Naszo2019}
Márton Naszódi.
\newblock Approximating a convex body by a polytope using the {Epsilon-Net
  Theorem}.
\newblock {\em Discrete \& Computational Geometry}, 61(3):686--693, 2019.
\newblock \href {https://doi.org/10.1007/S00454-018-9977-0}
  {\path{doi:10.1007/S00454-018-9977-0}}.

\bibitem{SchaeW1999}
Helmut~H. Schaefer and Manfred~P. Wolff.
\newblock {\em Topological Vector Spaces}.
\newblock Springer, 2nd edition, 1999.
\newblock \href {https://doi.org/10.1007/978-1-4612-1468-7}
  {\path{doi:10.1007/978-1-4612-1468-7}}.

\bibitem{SerraB1994}
Bruno Serra and Marc Berthod.
\newblock Subpixel contour matching using continuous dynamic programming.
\newblock In {\em 1994 Proceedings of IEEE Conference on Computer Vision and
  Pattern Recognition}, pages 202--207. IEEE, 1994.
\newblock \href {https://doi.org/10.1109/CVPR.1994.323830}
  {\path{doi:10.1109/CVPR.1994.323830}}.

\bibitem{SerraB1995}
Bruno Serra and Marc Berthod.
\newblock Optimal subpixel matching of contour chains and segments.
\newblock In {\em Proceedings of IEEE International Conference on Computer
  Vision}, pages 402--407. IEEE, 1995.
\newblock \href {https://doi.org/10.1109/ICCV.1995.466911}
  {\path{doi:10.1109/ICCV.1995.466911}}.

\bibitem{SuLZZZ2020}
Han Su, Shuncheng Liu, Bolong Zheng, Xiaofang Zhou, and Kai Zheng.
\newblock A survey of trajectory distance measures and performance evaluation.
\newblock {\em The VLDB Journal}, 29(1):3--32, 2020.
\newblock \href {https://doi.org/10.1007/S00778-019-00574-9}
  {\path{doi:10.1007/S00778-019-00574-9}}.

\bibitem{TaoBSBSPLPTD2021}
Yaguang Tao, Alan Both, Rodrigo~I. Silveira, Kevin Buchin, Stef Sijben, Ross~S.
  Purves, Patrick Laube, Dongliang Peng, Kevin Toohey, and Matt Duckham.
\newblock A comparative analysis of trajectory similarity measures.
\newblock {\em GIScience \& Remote Sensing}, 58(5):643--669, 2021.
\newblock \href {https://doi.org/10.1080/15481603.2021.1908927}
  {\path{doi:10.1080/15481603.2021.1908927}}.

\end{thebibliography}

    \appendix
    \section{Proof of \autoref{thm:apx-function-properties}}
\label{app:apx-function-properties-proof}

\settheoremcounter{thm:apx-function-properties}
\begin{proposition}
    For every border $\mathcal{B}$ the function $\mathcal{B}.\mathrm{apx}$ computed by our algorithm is continuous.
    Moreover, all $t,t' \in \mathrm{dom}(\mathcal{B})$ with $t < t'$ satisfy $\mathcal{B}.\mathrm{apx}(t) + \mathrm{opt}_{\| \cdot \|}(\mathcal{B}(t), \mathcal{B}(t')) \geq \mathcal{B}.\mathrm{apx}(t')$.
\end{proposition}

\begin{proof}
    We start with the second property.
    The initialisations of $\mathcal{B}.\mathrm{apx}$ in lines~1--2 and lines~6--7 of \autoref{alg:cdtw-apx} yield $\mathcal{B}.\mathrm{apx}(t) + \mathrm{opt}_{\| \cdot \|}(\mathcal{B}(t), \mathcal{B}(t')) = \mathcal{B}.\mathrm{apx}(t')$, as they use straight paths on $\mathcal{B}$.
    In $\mathcal{B}.\propagatefromadj$ from \autoref{alg:cdtw-subroutines}, where $\mathcal{A} := \mathcal{B}.\mathrm{adj}$, we have the bound
    \begin{align*}
        \mathrm{apx}_{s^*}(t) + \mathrm{opt}_{\| \cdot \|}(\mathcal{B}(t), \mathcal{B}(t'))
        &= \mathcal{A}.\mathrm{apx}(s^*) + \mathrm{opt}_{\mathcal{A},\mathcal{B}}(s^*,t) + \mathrm{opt}_{\| \cdot \|}(\mathcal{B}(t), \mathcal{B}(t')) \\
        &\geq \mathcal{A}.\mathrm{apx}(s^*) + \mathrm{opt}_{\mathcal{A},\mathcal{B}}(s^*,t')
        = \mathrm{apx}_{s^*}(t')
        \text{.}
    \end{align*}
    This is because we can extend an optimal $(\mathcal{A}(s^*),\mathcal{B}(t))$-path on $\mathcal{B}$ to create an $(\mathcal{A}(s^*),\mathcal{B}(t'))$-path whose cost is lower bounded by the cost of an optimal $(\mathcal{A}(s^*),\mathcal{B}(t'))$-path.
    For the same reason, we moreover have $\mathrm{apx}_{\leq s^*}(t) + \mathrm{opt}_{\| \cdot \|}(\mathcal{B}(t), \mathcal{B}(t')) \geq \mathrm{apx}_{\leq s^*}(t')$ in $\mathcal{B}.\propagatefromopp$.
    The subroutines compute $\mathcal{B}.\mathrm{apx}$ as the lower envelope of its initialisation with $\mathrm{apx}_{s^*}$ and $\mathrm{apx}_{\leq s^*}$, and the lower envelope retains the desired property that is shared by these three functions.

    To show continuity, we proceed similarly within an inductive proof.
    The base case in lines~1--2 of \autoref{alg:cdtw-apx} gives continuous cost functions by \autoref{thm:converging-costs} since the straight paths on~$\mathcal{B}$ converge.
    Consequently, lines~6--7 also initialise $\mathcal{B}.\mathrm{apx}$ as a continuous function.
    We further have that $\mathrm{apx}_{s^*}$ is continuous due to \autoref{thm:converging-costs} and the converging paths from \autoref{thm:optimal-paths}.
    It now remains to show that $\mathrm{apx}_{\leq s^*}$ is continuous as well.
    Then it follows that $\mathcal{B}.\mathrm{apx}$ is continuous as the lower envelope of a finite number of continuous functions on the same domain.
    We show $\limsup_{t \to t_0} \mathrm{apx}_{\leq s^*}(t) \leq \allowbreak \mathrm{apx}_{\leq s^*}(t_0) \leq \liminf_{t \to t_0} \mathrm{apx}_{\leq s^*}(t)$ for all $t_0 \in \mathrm{dom}(\mathcal{B})$, which implies continuity as required.

    Consider $\mathcal{A} := \mathcal{B}.\mathrm{opp}$ from now on, and assume inductively that $\mathcal{A}.\mathrm{apx}$ is continuous.
    For~each $t \in \mathrm{dom}(\mathcal{B})$ pick an $s^*_t \leq s^*$ that satisfies $\mathrm{apx}_{\leq s^*}(t) = \mathcal{A}.\mathrm{apx}(s^*_t) + \mathrm{opt}_{\mathcal{A},\mathcal{B}}(s^*_t,t)$, and let $\gamma^*_t$ be a corresponding optimal $(\mathcal{A}(s^*_t),\mathcal{B}(t))$-path.%
    \footnote{By choice of $s^*_t$, we have that $\gamma^*_t$ is a best path from $\mathcal{A}$ to $\mathcal{B}(t)$ as introduced in \autoref{def:best-path}. The existence of these values and paths follows from the continuity of the function $\mathcal{A}.\mathrm{apx}$, see the proof of \autoref{thm:best-paths}.}
    As above, we extend the paths $\gamma^*_t$ for $t \nearrow t_0$ to create $(\mathcal{A}(s^*_t),\mathcal{B}(t_0))$-paths, and we extend $\gamma^*_{t_0}$ up to $\mathcal{B}(t')$ for $t' \searrow t_0$.
    Using \autoref{thm:converging-costs} then yields
    \[
        \limsup_{\smash{t' \searrow t_0}} \, \mathrm{apx}_{\leq s^*}(t')
        \leq \mathrm{apx}_{\leq s^*}(t_0)
        \leq \liminf_{t \nearrow t_0\xmathstrut[0]{0.13}} \, \mathrm{apx}_{\leq s^*}(t)
        \text{,}
    \]
    where on the left side we can use it for the extended paths, which converge to $\gamma^*_{t_0}$, while on the right side we use it only for the suffix paths from $\mathcal{B}(t)$ to $\mathcal{B}(t_0)$, which vanish in the limit.

    In case of $s^*_{t_0} < t_0$, we can apply a different transformation to $\gamma^*_{t_0}$ for $t \nearrow t_0$:
    By concatenating straight paths onto the prefix paths of $\gamma^*_{t_0}$ up to the final point sharing a coordinate with each $\mathcal{B}(t)$, we obtain $(\mathcal{A}(s^*_{t_0}),\mathcal{B}(t))$-paths converging to $\gamma^*_{t_0}$.
    We then have $\limsup_{t \nearrow t_0} \mathrm{apx}_{\leq s^*}(t) \leq \mathrm{apx}_{\leq s^*}(t_0)$ due to \autoref{thm:converging-costs}.
    Applying this transformation to the paths $\gamma^*_{t'}$ for all $t' \searrow t_0$ with $s^*_{t'} < t_0$ yields $(\mathcal{A}(s^*_{t'}),\mathcal{B}(t_0))$-paths and $\mathrm{apx}_{\leq s^*}(t_0) \leq \liminf_{t' \searrow t_0 \;\! : \;\! s^*_{\smash{t'}\vphantom{t}} < t_0} \mathrm{apx}_{\leq s^*}(t')$ via \autoref{thm:converging-costs}, as now w.l.o.g.\ all the suffix paths are replaced with and converge to the same straight path by \autoref{thm:optimal-paths}.

    For all $t' \searrow t_0$ with $s^*_{t'} \geq t_0$ it follows $s^*_{t'} \searrow t_0$, so that the related paths $\gamma^*_{t'}$ converge to the straight $(\mathcal{A}(t_0),\mathcal{B}(t_0))$-path.
    This implies $\mathrm{apx}_{\leq s^*}(t_0) \leq \liminf_{t' \searrow t_0 \;\! : \;\! s^*_{\smash{t'}\vphantom{t}} \geq t_0} \mathrm{apx}_{\leq s^*}(t')$ by \autoref{thm:converging-costs} and continuity of $\mathcal{A}.\mathrm{apx}$.
    Finally, it remains to consider $t \nearrow t_0$ in case of $s^*_{t_0} = t_0$.
    We then have that $\gamma^*_{t_0}$ is a straight path.
    By shifting it to the left, we get $(\mathcal{A}(t),\mathcal{B}(t))$-paths that converge to it.
    Therefore, \autoref{thm:converging-costs} and continuity of $\mathcal{A}.\mathrm{apx}$ say $\limsup_{t \nearrow t_0} \mathrm{apx}_{\leq s^*}(t) \leq \mathrm{apx}_{\leq s^*}(t_0)$ here~too.
    \qed
\end{proof}

\section{A 3D Counterexample to \autoref{thm:optimal-paths}}
\label{app:3d-example}

Any line $\ell$ that induces optimal paths as in \autoref{thm:optimal-paths} is called a \emph{valley} of the cell $C$ under $\tnorm$.
Formally, for every valley point $z \in \ell \cap C$ the function $t \mapsto \| P_{\| \cdot \|}(z_1 + t) - Q_{\| \cdot \|}(z_2 - t) \|$ needs to be non-increasing on $\mathbb{R}_{\leq 0}$ and non-decreasing on $\mathbb{R}_{\geq 0}$, see \cite[Definition~4]{BuchiBSW2026}.
Intuitively, when approaching $\ell$ along any line of slope $-1$, the height of the cell terrain becomes smaller.

In 2D there always exists a valley of positive slope \cite[Theorem~8]{BuchiBSW2026}, but this does not hold in 3D already under the $1$-norm $\tnorm[1]$.
For example, consider the parameter space cell of single-segment curves $P := \langle (7,18,11)^\mathsf{T}, (28,25,25)^\mathsf{T} \rangle$ and $Q := \langle (12,36,18)^\mathsf{T}, (12,0,24)^\mathsf{T} \rangle$.
As seen in \autoref{fig:3d-example}, minima on lines of slope $-1$ are then attained on three line segments instead of on a single line.
Two of these segments even have negative slope, so monotone paths cannot travel on them.

\begin{figure}[H]
    \centering
    \includegraphics[width=0.75\linewidth]{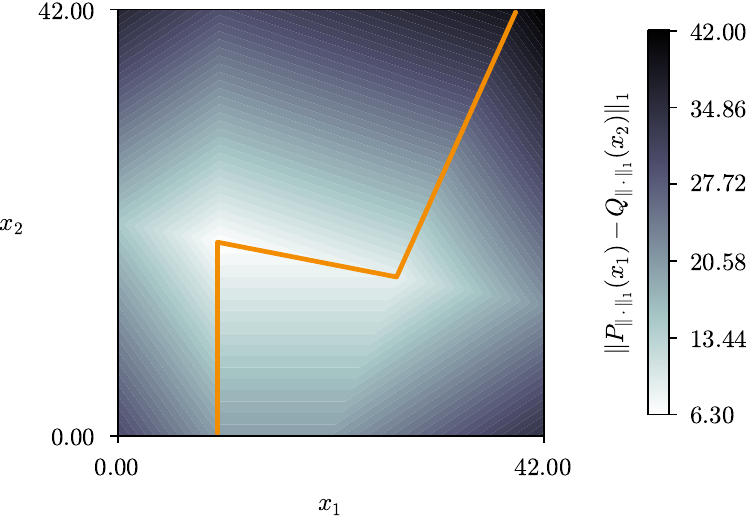}
    \caption{Parameter space cell of curve segments in $(\mathbb{R}^3, \tnorm[1])$ without a genuine valley}
    \label{fig:3d-example}
\end{figure}

It is not clear whether it is possible to extend our results to this kind of setting.
One approach might be to subdivide the cell into parts such that every part has a valley.
However, these~parts would not be rectangles anymore since we would need to use lines of slope $-1$ for the subdivision.
This could complicate propagations between borders.
Moreover, optimal paths that are induced by negatively sloped valleys are not as straightforward as those from \autoref{thm:optimal-paths}, see \cite[Lemma~27]{Brank2022}.
We leave such additional obstacles that appear in dimensions beyond 2D for future work.

\end{document}